\documentclass[11pt]{article}

\usepackage[english]{babel}
\usepackage[utf8]{inputenc}
\usepackage[authoryear]{natbib}
\usepackage{geometry}
\usepackage{adjustbox}
\usepackage{setspace}
\usepackage{fontawesome5} 
\usepackage{academicons}

\usepackage{amsmath}
\usepackage{amsfonts}
\usepackage{amsthm}
\usepackage{amssymb}
\usepackage{mathrsfs}
\usepackage{bbm}
\usepackage{amstext}
\usepackage{bm}
\usepackage{mathtools}
\usepackage{todonotes}

\theoremstyle{plain}

\theoremstyle{definition}
\newtheorem{remark}{Remark}

\newtheorem{example}{Example}

\newtheorem{proposition}{Proposition}

\DeclareMathOperator*{\argmin}{arg\,min}

\newcommand{\I}{\mathbbm{1}}

\numberwithin{equation}{section}
\usepackage{chngcntr}
\counterwithin*{section}{part}
\usepackage{apptools}

\usepackage{multirow}
\usepackage{babel}
\usepackage{array}
\usepackage{rotating}
\usepackage{longtable}
\usepackage{float}
\usepackage{booktabs}
\usepackage{lscape}

\providecommand{\tabularnewline}{\\}

\usepackage{caption}
\usepackage{subcaption}

\usepackage[ruled]{algorithm2e}

\usepackage{color}
\usepackage{graphicx}

\usepackage{xcolor}
\usepackage{listings}
\usepackage{alltt}

\definecolor{codegreen}{rgb}{0,0.6,0}
\definecolor{codegray}{rgb}{0.5,0.5,0.5}
\definecolor{codepurple}{rgb}{0.58,0,0.82}
\definecolor{backcolour}{rgb}{0.95,0.95,0.92}

\lstdefinestyle{mystyle}{
    backgroundcolor=\color{backcolour},   
    commentstyle=\color{codegreen},
    keywordstyle=\color{magenta},
    numberstyle=\tiny\color{codegray},
    stringstyle=\color{codepurple},
    basicstyle=\ttfamily\footnotesize,
    breakatwhitespace=false,         
    breaklines=true,    captionpos=b,                    
    keepspaces=true, numbers=left,                    
    numbersep=5pt, showspaces=false,                
    showstringspaces=false, showtabs=false, tabsize=2}
\usepackage{url}
\usepackage{hyperref}
\hypersetup{
    colorlinks=true,
    linkcolor=[rgb]{0.6000000, 0.1058824, 0.1176471}, 
    filecolor=[rgb]{0.6000000, 0.1058824, 0.1176471}, 
    urlcolor=[rgb]{0.6000000, 0.1058824, 0.1176471},
    citecolor=[rgb]{0.6000000, 0.1058824, 0.1176471}, 
    linkbordercolor=[rgb]{1.0000000, 0.7803922, 0.1725490}, 
    citebordercolor=[rgb]{0.8666667, 0.6392157, 0.0000000}, 
    filebordercolor=[rgb]{1.0000000, 0.7803922, 0.1725490}, 
    urlbordercolor=[rgb]{1.0000000, 0.7803922, 0.1725490} 
}
\definecolor{myred}{rgb}{0.6000000, 0.1058824, 0.1176471}
\definecolor{mypurple}{rgb}{0.4588235, 0.05882352, 0.4274509}

\usepackage{todonotes}

\usepackage{rotfloat}
\usepackage{pdfpages}
\usepackage{indentfirst}
\usepackage{enumitem}

\makeatletter
\def\maxwidth{ %
  \ifdim\Gin@nat@width>\linewidth
    \linewidth
  \else
    \Gin@nat@width
  \fi
}
\makeatother

\definecolor{fgcolor}{rgb}{0.345, 0.345, 0.345}

\usepackage{framed}
\makeatletter
 {\par\unskip\endMakeFramed%
 \at@end@of@kframe}
\makeatother

\definecolor{shadecolor}{rgb}{.97, .97, .97}
\definecolor{messagecolor}{rgb}{0, 0, 0}
\definecolor{warningcolor}{rgb}{1, 0, 1}
\definecolor{errorcolor}{rgb}{1, 0, 0}

\allowdisplaybreaks

\title{
    \Large Robust Estimation of Regression Models with Potentially Endogenous Outliers via a Modern Optimization Lens 
}

\author{
    {Zhan Gao}\thanks{
       Department of Economics, University of Southern California.
       \texttt{zhangao@usc.edu}.
    } 
    \and 
    Hyungsik Roger Moon\thanks{
        Department of Economics, University of Southern California.
        \texttt{moonr@usc.edu}.
    }
}

\begin{document}
\maketitle
\begin{abstract}
    This paper addresses the robust estimation of linear regression models in the presence of potentially endogenous outliers. Through Monte Carlo simulations, we demonstrate that existing $L_1$-regularized estimation methods, including the Huber estimator and the least absolute deviation (LAD) estimator, exhibit significant bias when outliers are endogenous. Motivated by this finding, we investigate $L_0$-regularized estimation methods. We propose systematic heuristic algorithms, notably an iterative hard-thresholding algorithm and a local combinatorial search refinement, to solve the combinatorial optimization problem of the \(L_0\)-regularized estimation efficiently. Our Monte Carlo simulations yield two key results: (i) The local combinatorial search algorithm substantially improves solution quality compared to the initial projection-based hard-thresholding algorithm while offering greater computational efficiency than directly solving the mixed integer optimization problem. (ii) The $L_0$-regularized estimator demonstrates superior performance in terms of bias reduction, estimation accuracy, and out-of-sample prediction errors compared to $L_1$-regularized alternatives. We illustrate the practical value of our method through an empirical application to stock return forecasting.
\end{abstract}

 \thispagestyle{empty}
\newpage
\setcounter{page}{1}

\section{Introduction\label{sec:intro}}

Robust estimation of linear regression models in the presence of outliers has been extensively studied in econometrics and statistics. It is well-known that the classical ordinary least square (OLS) estimator is extremely sensitive to outliers. Various robust estimation methods have been proposed, for example, Huber's M-estimation \citep{huber1964robust}, the least median of squares (LMS) estimator \citep{siegel1982robust}, and the least trimmed squares (LTS) estimator \citep{rousseeuw:1984:lms_lts,rousseeuw1985multivariate}, among others. More recently, 
\citet{lee2012regularization} propose regularization of the \(L_1\)-norm of case-specific parameters to achieve robust estimation against outliers and \citet{she2011outlier} show that any regularized estimator can be formulated as an equivalent iterative thresholding procedure of outlier detection. Refer to \citet{yu2017robust} for a comprehensive survey and detailed comparison.

The traditional literature has focused on the breakdown point and efficiency of a robust estimator \citep{huber1983notion}. The finite sample breakdown point of an estimator measures the proportion of outliers that can be arbitrarily contaminated before the estimation error goes to infinity. On the other hand, efficiency measures the relative estimation efficiency of the robust estimator compared to OLS in the ideal scenario where the error term is normally distributed and there are no outliers. However, there is limited work investigating how robust estimators perform when the outliers are \textit{endogenous} in the sense that the noises can be arbitrarily correlated with observed regressors. In this scenario, an outlier not only brings in a contamination noise of large magnitude but also introduces model misspecification against the classical linear regression assumptions. The endogeneity issues from the outliers can lead to severe estimation bias.

In this paper, we first examine the widely adopted \(L_1\)-regularized estimator, which includes the Huber estimator and least absolute deviation as special cases, and demonstrate that it is subject to significant bias when outliers are endogenous. As shown in \citet{she2011outlier}, the \(L_1\)-regularized estimation is equivalent to an iterative soft-thresholding procedure with a data-driven thresholding parameter, and hence it does not completely eliminate the detected outlier. 

Motivated by this finding, we turn to the \(L_0\)-regularized approach. By restricting the cardinality of the set of outliers, the \(L_0\)-regularized estimator searches for the best subset of observations and the regression coefficients that minimize the least squares objective function. This problem is equivalent to solving the least trimmed squares (LTS) estimator, which is well-known to be an NP-hard combinatorial optimization problem \citep{natarajan:1995:nphard}. Following \citet{bertsimas/king/mazumder:2016:bestsubsetmio,thompson2022robust}, the \(L_0\)-regularized estimator can be formulated as a mixed integer optimization (MIO) problem. With the continuous advancement in modern optimization software, such as \texttt{gurobi}\footnote{\texttt{gurobi} is a commercial optimization solver that specializes in mixed integer optimization. In the past decade, \texttt{gurobi} has been up to 75-fold faster in computation speed for mixed integer optimization problems independent of hardware advancement \citep{gurobi}.}, and computational infrastructure, it is now tractable to solve the \(L_0\)-regularized estimation problem with a verifiable global optimal solution for datasets with up to hundreds of observations. However, the optimization routine heavily relies on the initial values because of its non-convexity nature. In addition, the implementation is not scalable, and the computational burden is increasingly heavy as the sample size and fraction of outliers grow.

To address the computational challenges, we propose a scalable and efficient heuristic algorithm that combines the iterative hard-thresholding (IHT) algorithm and a local combinatorial search refinement inspired by \citet{hazimeh/mazumder:2020:fastsubset}. Given the initial solution by the IHT algorithm, the local combinatorial search algorithm checks whether swapping a few observations, say one or two, between the estimated set of inliers and outliers, can improve the objective value by solving a small-scale mixed integer optimization problem. The heuristic algorithm is more computationally efficient and can improve the initial IHT solution to be as good as the solution from solving the original \(L_0\)-regularized MIO problem in most cases, according to the Monte Carlo experiments.

Our contributions are twofold. First, we document that existing \(L_1\)-regularized methods are biased in the presence of endogenous outliers, whereas the \(L_0\)-regularized approach does not suffer from this bias through Monte Carlo simulations. This finding extends the understanding of the properties of robust estimation methods to a new scenario. Second, we propose systematic heuristic algorithms that provide stable and high-quality solutions while being computationally efficient. To the best of our knowledge, this is the first work to apply the idea of local combinatorial search to the context of robust estimation and outlier detection.

To illustrate the practical value of our method, we apply it to an empirical application in stock return forecasting. The results demonstrate that our \(L_0\)-regularized estimator outperforms \(L_1\)-regularized alternatives in terms of out-of-sample prediction errors.

\bigskip

\noindent {\it Notations}. For a generic vector $a =  \left( a_1, a_2, \cdots, a_N \right)^{\prime}$, $\left\Vert a \right\Vert = \left( a^{\prime} a \right)^{1/2}$, $ \left\Vert a \right\Vert_1 = \sum_{i=1}^N \left\vert a_i \right\vert $, \(\left\lVert a \right\rVert_\infty = \max_i \vert a_i \vert \)  and \(\left\lVert a \right\rVert _{0} = \sum_{i=1}^{N} \I \left\{ a_i \neq 0 \right\}  \). 
For matrix $A$, we define the Frobenius matrix norm 
$\left\Vert A  \right\Vert$ as $\left\Vert A \right\Vert = \left( \operatorname{tr}\left( A^{\prime}A \right) \right)^{1/2}$. 
Generically, $[N] = \left\{ 1, 2, \cdots, N \right\}$ for positive integer $N$.
For \(a \in \mathbb{R} \), \(
    \lfloor a \rfloor   
\) is the integer part of the real number \(a\).
\(\iota_N = \left( 1,1,\cdots, 1 \right)\in \mathbb{R} ^N \) denotes the one vector. 

\bigskip

The rest of the paper is organized as follows. Section \ref{sec:model} presents the model setup
and motivating examples. Section \ref{sec:L1_reg} summarizes the existing \(L_1\)-regularized methods for robust estimation. Section \ref{sec:L0_reg_est} delineates the \(L_0\)-regularized robust methods and the detailed algorithms. The performance of the algorithms is evaluated through Monte Carlo experiments in \ref{sec:mc}. Section \ref{sec:empirical} illustrates the proposed method in an empirical application of stock return forecasting. Section \ref{sec:conclusion} concludes.

\section{Linear Regression with Potentially Endogenous Outliers\label{sec:model}}
This section sets up the linear regression models with potentially endogenous outliers. Suppose we observe samples $(y_i,x_i')'$ for individual unit $i \in 
\left[ N \right] $. 
Let $\mathcal{O}\subset \left[ N \right]$ and   
$\mathcal{I} = \left[ N \right]  \setminus \mathcal{O}$. 
Consider the following linear regression model:
\begin{equation}
	y_i = \beta_0 + x_i^\prime \beta_{1} + \alpha_i + u_i \label{eq:model.outliers}
\end{equation}
where $\mathbb{E}(u_i| x_i) = 0$ for all $i \in \left[ N \right] $. The parameter $\alpha_i$ represents the conditional mean shift such that 
\begin{align*}
	\alpha_i &= 0 \quad {\rm if} \,\,  i \in \mathcal{I} \\
	\mathbb{E}(\alpha_i | x_i) &\neq 0 \quad {\rm if} \,\,  i \in \mathcal{O}. 
\end{align*}
We refer to the parameters $\alpha_i$ as outlier fixed effects, which are allowed to be arbitrarily correlated with the regressors $x_i$. 
An alternative formulation of model \eqref{eq:model.outliers} is:
\begin{equation}
    \label{eq:model.outliers.with.gamma}
    y_i = \beta_0 + x_i^\prime \beta_{1} + \gamma_i \widetilde{\alpha} _i + u_i,
\end{equation}
where $\gamma_i = 
\I\left\{ i \in \mathcal{O} \right\}$ is the outlier dummy and \(\widetilde{\alpha} \) is the latent outlier fixed effect. 
\begin{remark}
    $\gamma_i$ and $\alpha_i$ are generally allowed to be correlated. For example, $\gamma_i = \I\left( \widetilde{\alpha} _i > \overline{a}  \right)$, 
    i.e., the shock \(\widetilde{\alpha}_{i}  \) affects the dependent variable only if it is large enough to pass the threshold \(\overline{a} \). Conversely, $\gamma_i = \I\left( \widetilde{\alpha} _i < \overline{a}  \right)$ implies that the error affects the dependent variable only if it is small enough to avoid detection during the data collection process.
\end{remark}
The existence of \(\alpha _{i} \) introduces endogeneity to a subset of observations \(\mathcal{O} \). The classical ordinary least squares (OLS) estimator is expected to be biased in the presence of endogenous outliers. The following motivating examples demonstrate possible sources of endogenous outliers.

\begin{example}[Heterogeneous Coefficients] 
Suppose that for $i \in \mathcal{I}$, 
\[
	y_i = \beta_0  + \beta_1' x_i + u_i,
\]
while for $i \in \mathcal{O}$, 
\begin{equation*}
	y_i = \beta_{0,i} + \beta_{1,i}' x_i + u_i,
\end{equation*}
where the coefficients $\beta_{0,i}$ and $\beta_{1,i}$ are heterogeneous across \(i \in \mathcal{O}\) and potentially correlated with \(x_i\). 
Define 
\begin{align*} 	
	\alpha_{0,i} &:= (\beta_{0,i} - \beta_0) \I( i \in \mathcal{O}), \\
	\alpha_{1,i} &:= (\beta_{1,i} - \beta_1) \I( i \in \mathcal{O}).
\end{align*}
Then, for $i \in \left[ N \right]$,
\begin{equation*}
        y_i = \beta_0 + \beta _1 ^{\prime} x _{i} + \alpha_{0,i} + \alpha_{1,i}^{\prime}  x_i + u_i, 
        \label{eq:hetero_coef_example}
\end{equation*} 
which can be reformulated as \eqref{eq:model.outliers} with 
\begin{equation*}
    \alpha_i = \begin{cases}
        0 & {\rm for} \;\; i \in \mathcal{I}, \\
        \alpha_{0, i} +  \alpha_{1,i} ^{\prime} x_i & {\rm for} \;\; i \in 
        \mathcal{O},
    \end{cases}
\end{equation*}
where \(\mathbb{E} \left( \alpha_i | x_i \right) = \mathbb{E} \left( 
    \alpha_{0,i} | x_i \right)  + \mathbb{E} \left( \alpha_{1,i} | x_i \right) ^{\prime} x _{i} 
\neq 0 \) for \(i \in \mathcal{O}\) in general.

\end{example}
\begin{example}[Measurement Errors]
	Suppose that $y_i$ and $x_i^*$ are the outcome variable and the regressors 
	of interest, respectively. Assume that there exists a linear relationship 
	between $y_i$ and $x_i^*$ as
		\[
			y_{i} = \beta_0 + \beta_1'x_i^* + u_i,
		\]
	where $\mathbb{E}(u_i | x_i^*) = 0$.
	However, measurement errors can contaminate the variables during data collection and processing. We observe 
	$x_i^*$ only for $i \in \mathcal{I}$, and for $i \in \mathcal{O}$, we 
	observe proxy variables with measurement errors, $x_i := x_i^* + 
	\alpha_{1,i}$, where $\alpha_{1,i}$ represents measurement errors. 
	Then, the regression model becomes 
    \begin{equation*}
        y_i = \beta_0 + x_i^\prime \beta_1 + \alpha_i + u_i,
    \end{equation*}
    where \begin{equation*}
        x_i = \begin{cases}
            x_i^* & {\rm for}\;\; i \in \mathcal{I},\\
            x_i^* + \alpha_{1,i} & {\rm for}\;\; i \in \mathcal{O},
        \end{cases}\quad \quad \alpha_i = \begin{cases}
            0 & {\rm for}\;\; i \in \mathcal{I},\\
            -\alpha_{1,i}^\prime \beta_1 & {\rm for}\;\; i \in \mathcal{O}.
        \end{cases} 
    \end{equation*}
    In the case of nonclassical measurement errors, $\mathbb{E}(\alpha_i | x_i) \neq 0$ in general, which induces endogeneity to a subset of observations \(\mathcal{O} \).
\end{example}

The objective is to accurately estimate $\beta = (\beta_0,\beta_1')'$ in the presence of outliers in the samples.

\section{Existing Methods: $L_1$-regularized Robust Regression\label{sec:L1_reg}}

The most widely used robust estimators of $\beta$ in the presence of outlier observations are Huber's M-estimation and the least absolute deviation (LAD) estimation method \citep{huber1964robust}. These two estimators can be understood as special cases of the $L_1$-regularized estimator. 

Denote $Y = (y_1,...,y_N)', X = \left((1,x_1'),...,(1,x_N')\right)', \beta = \left( \beta_0, \beta_1^\prime \right)^\prime, \alpha = (\alpha_1,...\alpha_N)', U = (u_1,...,u_N)'$. Then, we can present the model (\ref{eq:model.outliers}) in matrix form as 
\begin{equation}
    \label{eq:model.outliers.matrix}
    Y = X\beta + \alpha + U.
\end{equation}
Let the least squares loss function be
\begin{equation*}
    L_N \left( \beta, \alpha \right) = \frac{1}{2} \left\Vert Y - X\beta - \alpha \right\Vert ^2,
\end{equation*} 
and consider the following $L_1$-regularized optimization problem,
\begin{equation}
    \min_{\beta, \alpha} Q_N^\psi\left( \beta, \alpha \right) = L_N\left( 
    \beta, \alpha \right) + \psi \left\Vert \alpha \right\Vert_1. \label{eq:l1_prob}
\end{equation}
Given $\beta$, \(
    \hat{\alpha}^{\psi}(\beta) \coloneq \argmin_\alpha Q_{N}^{\psi}(\beta, 
    \alpha).
\) 
The closed-form solution is
\begin{equation}
    \hat{\alpha}_{i}^{\psi}(\beta) = \begin{cases}
        \left( y_i - x_i^\prime \beta \right) - \psi & \text{if }  y_i - 
        x_i^\prime \beta \geq \psi, \\
        0 & \text{if } \left\vert  y_i - x_i^\prime \beta \right\vert < \psi, \\
        \left( y_i - x_i^\prime \beta \right) + \psi & \text{if }  y_i - 
        x_i^\prime \beta \leq -\psi.
    \end{cases} \label{eq:closed.form.alpha}
\end{equation}
Then the profile $L_1$-regularized objective function becomes
\begin{equation}
        Q_N^\psi\left( \beta, \hat{\alpha}^\psi\left( \beta \right) \right)  = 
    \sum_{i=1}^{N}\left[\frac{1}{2}\left(y_{i}-\beta^{\prime} x_{i}\right)^{2} 
    \I\left\{\left|y_{i}-\beta^{\prime} x_{i}\right| \leq 
    \psi\right\}+\left(\psi\left|y_{i}-\beta^{\prime} 
    x_{i}\right|-\frac{\psi^{2}}{2}\right) 
    \I\left\{\left|y_{i}-\beta^{\prime} 
    x_{i}\right|>\psi\right\}\right].\label{eq:L1-objective}
\end{equation}

\begin{figure}[h]
    \centering
    \includegraphics[width=0.5\textwidth]{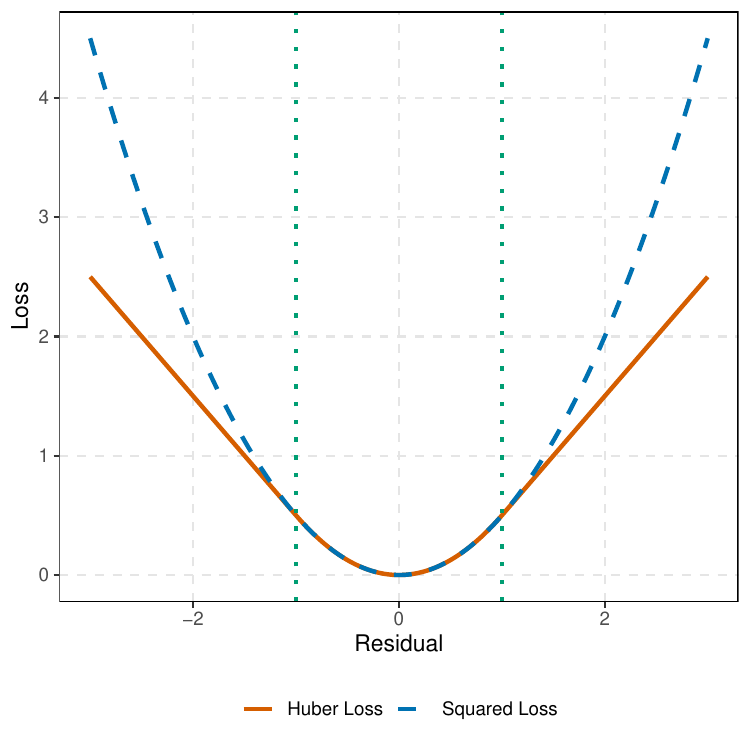}
    \caption{The $L_1$-regularized loss function\label{fig:L1-loss}}
\end{figure}

In \eqref{eq:closed.form.alpha}, the tuning parameter \(\psi \) 
plays the role of a soft thresholding parameter. 
When \(\psi\)  is fixed, \eqref{eq:L1-objective} becomes the objective function of Huber regression
\begin{equation*}
    Q_N^\psi\left( \beta, \hat{\alpha}^\psi\left( \beta \right) \right)  = \sum_{i=1}^{N} \rho_{\psi} \left( y_i - \beta^\prime x_i \right),
\end{equation*}
where \begin{equation}
    \rho_{\psi}(t) = \begin{cases}
      \frac{1}{2} t^2 & \text{if } |t| \leq \psi, \\
      \psi |t| - \frac{1}{2} \psi^2 & \text{if } |t| > \psi,
    \end{cases} \label{eq:huber_loss}
\end{equation}
is the Huber loss function \citep[p. 26 - 27]{hastie/tibshirani/wainwright2015}. 
As illustrated in Figure \ref{fig:L1-loss}, the cutoffs \(-\psi \) and  \(\psi \) 
divide the Huber loss function into two regimes: 
it is a squared loss if \(|t| \leq \psi \) and a linear loss in absolute deviation if \(|t| > \psi \). 
If we let \(\psi = \psi _N\) so that $\min_{i\in \left\{ 1,2,\cdots, N \right\}} 
\left\vert y_i - x_i^{\prime} \beta \right\vert \geq \psi_{N}  > 0$ given \(N\) and the sample \(\left( Y, X\right) \) , then we have 
\begin{equation*}
    \left( N\psi_{N}  \right)^{-1} Q_N^\psi\left( \beta, \hat{\alpha}^\psi \left( 
    \beta \right) \right) \to \frac{1}{N} \left\vert y_i - x_i^\prime \beta 
    \right\vert,
\end{equation*}
as $\psi_{N} \to 0$, which shrinks the squared loss region in Figure \ref{fig:L1-loss}.

The $L_1$-regularized estimator unifies the Huber and LAD estimators and offers greater flexibility by allowing the regularization parameter to be chosen in a data-driven manner. However, as a soft-thresholding algorithm, as shown in \eqref{eq:closed.form.alpha}, it does not completely eliminate the detected outliers from the estimation procedure. This can lead to estimation biases when the outliers are endogenous. 

As demonstrated in the Monte Carlo experiments in Section \ref{sec:mc}, the $L_1$-regularized estimator performs well in terms of estimation bias, root mean squared error (RMSE), and prediction error in DGP 1, where the outlier fixed effects are generated as exogenous random variables. However, in DGP 2 and DGP 3, where the outlier fixed effects are correlated with the regressors, the $L_1$-regularized estimator for the slope coefficients exhibits significant biases. This finding motivates the consideration of the hard-thresholding-based $L_0$-regularized estimator, as discussed in Section \ref{sec:L0_reg_est}.

\section{$ L_0 $-regularized Robust Regression\label{sec:L0_reg_est}}
We consider the $L_0$-reguarization on the outlier fixed effects in the least square estimation,
\begin{equation}
    \label{eq:L0_prob}
    \min_{\beta, \alpha} L_N\left( \beta, \alpha \right) \text{ s.t. } \left\Vert \alpha \right\Vert_0 \leq k, 
\end{equation}
where \(
    \left\lVert \alpha  \right\rVert _{0}  = \sum_{i=1}^{N} \I\left\{ \alpha _{i} \neq 0 \right\} 
\) and the tuning parameter $k\in \left[ N \right] $ controls the exact sparsity of the outlier fixed effects $ \alpha $.

\begin{remark}
    The Lagrangian form of the regularized estimation problem 
    \begin{equation}
        \label{eq:penalized_ls_prob}
        \min_{\beta, \alpha} L_N\left( \beta, \alpha \right) + P_\psi\left( \alpha \right),
    \end{equation}
    where $ P_\psi \left( \alpha \right) $ is the penalty function, is often studied in the literature, for example, by \citet{she2011outlier} and \citet{lee2012regularization}. When $ P_\psi \left( \alpha \right) = \psi\left\Vert \alpha \right\Vert_1 $, \eqref{eq:penalized_ls_prob} is equivalent to classical Huber regression or least absolute deviation (LAD) estimation, depending on the tuning parameter $\psi$, which is detailed in Section \ref{sec:L1_reg}. \citet{she2011outlier} show that the iterative procedure of outlier detection based on hard-thresholding, $ \Theta_{\operatorname{hard }}^{\psi } (x)= \begin{cases}0, & \text { if }|x| \leq \psi  \\ x, & \text { if }|x|> \psi \end{cases} $, gives a local minimum of the $L_0$-regularized optimization,
    \begin{equation}
        \label{eq:L0_penalize}
        \min_{ \beta, \alpha} L_N\left( \beta, \alpha \right) + \psi \left\Vert  \alpha \right\Vert_0.
    \end{equation}
    Note that \eqref{eq:L0_prob} and \eqref{eq:L0_penalize} are not equivalent since it is not guaranteed that there exists a corresponding $ \psi > 0 $ for each $ k \in \left[ N \right]  $. To see this, denote $ \psi \left( k \right) =  \min_{
        \left\Vert \beta \right\Vert_0 \leq k
    } \frac{1}{2} \left\Vert y -  X\beta \right\Vert_2^2$, and 
    $ 
    \delta\left( k \right) = \left\vert 
        \psi \left( k + 1 \right) - \psi \left( k \right)
    \right\vert 
    $. Note that $ \psi \left( k \right) $ is strictly decreasing in $ k $ unless the model perfectly fits the data. Suppose $ \delta \left( k \right) $ is decreasing, then for any $ k $, there exists $ \psi \in \left( \delta \left( k + 1 \right), \delta \left( k \right) \right) $ such that the corresponding $ L_0 $-penalized problem gives the same solution as \eqref{eq:L0_prob}. However, $ \delta \left( k \right) $ is not guaranteed to be decreasing. For any local minimum $ \tilde{k} $, there is no corresponding $ \psi \in \mathbb{R}_{+} $ that yields the same solution. Scenarios like cointegration can serve as examples. In this paper, we focus on the constrained form of $ L_0 $-regularization \eqref{eq:L0_prob}, which allows one to control the exact sparsity of $ \alpha $.
\end{remark}

As in \citet{bertsimas/king/mazumder:2016:bestsubsetmio,thompson2022robust}, \eqref{eq:L0_prob} can be formulated as the following mixed integer optimization (MIO) problem,
\begin{align}
    & \min_{\beta, \alpha, \gamma} L_N\left( \beta, \alpha \right)\nonumber \\
    \text{s.t. } &  
    \gamma_i \in \left\{ 0, 1 \right\},\, i \in \left[ N \right] \nonumber \\
    & \sum_{i=1}^N \gamma_i\leq k,\label{eq:mio_prob} \\
    & \left\{\alpha_i, 1 - \gamma_i  \right\} \in \text{SOS-1},\, i \in \left[ N \right],  \nonumber\\
    & \left\lVert \alpha  \right\rVert_{\infty} \leq   M_{\alpha},\, \left\lVert \alpha \right\rVert_1 \leq M_{\alpha, 1}, \nonumber
\end{align}
in which we introduce a binary variable \(\gamma _{i} \) to model whether an observation is detected as an outlier. \(\sum_{i=1} ^{N} \gamma _{i} \leq k \) corresponds to the cardinality constraint \(\left\lVert \alpha  \right\rVert _0 \leq k \). If \(\gamma _{i}  = 1\), i.e. \(i\) is labelled as an outlier, then \(\alpha _{i} \) should be nonzero; on the other hand if \(\gamma _{i} = 0\), i.e. \(i \) is labelled as an inlier, then \(\alpha _{i} \) should be exactly 0. This intuition can be translated into the constraint  \(\left( 1 - \gamma _{i}  \right)\alpha _{i} =0 \), which can be further modeled via integer optimization using the Special Ordered Set of Type 1 (SOS-1), that is a set contains at most one non-zero variable.\footnote{
    \eqref{eq:L0_prob} can also be formulated as a MIO problem using big-\(M\)  method,
    \begin{equation*}
        \label{eq:mio_bigm}
        \min_{\beta, \alpha, \gamma} L_N\left( \beta, \alpha \right) 
        \text{ s.t. }
        \gamma_i \in \left\{ 0, 1 \right\},\,
        \sum_{i=1}^N \gamma_i\leq k,\,
        -\gamma_i M_\alpha \leq \alpha_i \leq \gamma_i M_\alpha, 
    \end{equation*}
    where $M_\alpha$ is provides a sufficiently large bound for $ \alpha $. 
    In practice, we set \(M_{\alpha} = \tau \left\lVert \widehat{\alpha}  \right\rVert_{\infty} \) with \(\tau > 1\) and \(\widehat{\alpha} \) is the initial estimator form the heuristic algorithm.   
} \(M_{\alpha } \) and \(M_{\alpha, 1}\) are user-defined bound parameters that can help to tighten the parameter space and improve the computation performance. 

The optimization problem \eqref{eq:mio_prob} is a well-posed MIO problem, and provable global optimality can be achieved by optimization solvers such as \texttt{gurobi} \citep{gurobi}. The computational efficiency and solution quality are highly dependent on the initial values provided to the solver. Furthermore, as demonstrated in the Monte Carlo experiments in Table \ref{tab:algo_compare}, directly solving \eqref{eq:mio_prob} does not achieve provable optimality within the prespecified 5-minute timeframe when the sample size \(N\) and the number of outliers \(k\) are large. In the following subsection, we propose systematic approximate algorithms to solve \eqref{eq:L0_prob}.

\subsection{Heuristics\label{subsec:heuristics}}

The same as best subset selection studied in \citet{bertsimas/king/mazumder:2016:bestsubsetmio}, $ L_0 $-regularization robust regression \eqref{eq:L0_prob} cannot be solved in polynomial time, i.e. it is NP-hard \citep{natarajan:1995:nphard}. The computational cost of solving for the exact solution increases steeply as the sample size increases. Heuristic algorithms that can find approximate solutions are useful for parameter tuning and providing warm-starts for the solvers. In this section, we follow \citet{bertsimas/king/mazumder:2016:bestsubsetmio,hazimeh/mazumder:2020:fastsubset,thompson2022robust,mazumder/radchenko/dedieu:2022:subsetselection} to provide heuristic algorithms based on iterative hard-thresholding, neighborhood search, local combinatorial search for \eqref{eq:L0_prob}.

\subsubsection{Iterative Hard-thresholding}

Iterative hard-thresholding (IHT) based on project gradient descent is a generic heuristic algorithm for general nonlinear programming with sparsity constraints \citep{beck/eldar:2013:sparsityconstraint} and it is successfully applied to best subset selection in \citet{bertsimas/king/mazumder:2016:bestsubsetmio}. For problem \eqref{eq:L0_prob}, we propose the following iterative hard-thresholding algorithm, which is a simplified version of the projected block-coordinate gradient descent in \citet[Algorithm 1]{thompson2022robust}.

Define the hard-thresholding operator for $ c\in \mathbb{R}^{N} $  as
\begin{equation*}
    H_k\left( c \right) = \argmin_{\left\Vert \alpha \right\Vert_0 \leq k} \left\Vert \alpha - c \right\Vert_2^2.
\end{equation*}
As shown in \citet[Proposition 3]{bertsimas/king/mazumder:2016:bestsubsetmio}, if $ \hat{\alpha} \in H_k\left( c \right)$, then $ \hat{\alpha} $  retains the $ k $ largest (in absolute value) elements of $ c $ and sets the rest to zero, i.e.
\begin{equation*}
    H_k\left( c \right) = \begin{cases}
        c_i &  \text{if } i \in \left\{ (1), (2), \cdots, (k) \right\}\\
        0 & \text{otherwise},
    \end{cases}
\end{equation*}
if $ \left\{ (1), (2), \cdots, (N) \right\}$ is an ordering of $ \mathcal{N} $ such that $ \left\vert c_{(1)} \right\vert \geq \left\vert c_{(2)} \right\vert \geq \cdots \geq \left\vert c_{(N)} \right\vert$.

The iterative hard-thresholding algorithm starts with an initial robust estimator of $ \beta $ and iteratively updates $ \alpha $ and $ \beta $ by applying the hard-thresholding operator to the residuals and running least square estimation based on the support of updated $ \alpha $. The details are summarized in Algorithm \ref{algo:iht}.

\begin{algorithm}
    \caption{Iterative Hard-thresholding (IHT)}
    \label{algo:iht}
    \SetKwInOut{Input}{input}
    \SetKwInOut{Output}{output}

    \Input{
        An initial robust estimator $\hat{\beta}^{(0)}$, sparsity parameter $k$ and the maximum number of iterations $M$.
    }
    Initialize $j = 0$,
    \While{
        $\left\Vert \hat{\beta}^{(j+1)} - \hat{\beta}^{(j)} \right\Vert_2 > 0$ and $j\leq M$
    }{
        $\hat{\alpha}^{(j)} \gets H_k\left( Y - X\hat{\beta}^{(j)} \right)$\;
        $\hat{\beta}^{(j+1)} \gets \arg\min_{\beta} L_N\left( \beta, \hat{\alpha}^{(j)} \right)$\;
        $j\gets j+1$
    }
    \Output{
        $\hat{\beta}^{\text{IHT}} = \hat{\beta}^{(j)}$ and $\hat{\alpha}^{\text{IHT}} = \hat{\alpha}^{(j)}$.
    }
\end{algorithm}

\begin{remark}
    The initial estimator \(\hat{\beta}^{(0)}\) for Algorithm \ref{algo:iht} can be obtained by Huber regression or least absolute deviation estimation.
    Algorithm \ref{algo:iht} is a simplified version of the projected block-coordinate gradient descent in \citet[Algorithm 1]{thompson2022robust}. The convergence of IHT directly follows from \citet{thompson2022robust}. 
    In addition, $ \left( \hat{\beta}^{\text{IHT}^\prime}, {\hat{\alpha}^{\text{IHT}^\prime}}  \right)^{\prime}$  is automatically coordinate-wise optimal in the sense that optimizing concerning one coordinate at a time, while keeping others ﬁxed, cannot improve the objective \citep{hazimeh/mazumder:2020:fastsubset}.
\end{remark}

\subsubsection{Local Combinatorial Search}


Despite the iterative hard-thresholding algorithm is fast and easy to implement, it is only guaranteed to converge to a coordinate-wise optimal solution. In this section, we propose a local combinatorial search algorithm, similar to \citet{hazimeh/mazumder:2020:fastsubset}, to further refine the solution.

Note that \eqref{eq:L0_prob} is equivalent to 
\begin{equation*}
    \min_{
        \mathcal{I}\subset \left[ N \right],\, 
        \left\vert \mathcal{I} \right\vert \geq  N - k
    } \min_{\beta} 
    \sum_{i\in \mathcal{I}} \left( y_i - x_i^\prime \beta \right)^2.
\end{equation*}
For an estimate $ \hat{\alpha} $, denote $ \hat{\mathcal{I}} = \left\{ i\in \mathcal{N}: \hat{\alpha}_i = 0 \right\}$ and $ \hat{\mathcal{O}} = \left\{ i\in \mathcal{N}: \hat{\alpha}_i \neq 0 \right\}$. Following \citet{hazimeh/mazumder:2020:fastsubset},  $ \hat{\alpha} $ is said to be swap-inescapable of order $ l $ if arbitrarily swapping up to $ l $ observations between $ \hat{\mathcal{I}} $ and $ \hat{\mathcal{O}} $ and then optimizing over the new support cannot improve the objective value. Formally,
\begin{equation}
    \label{eq:local_comb_search_prob}
    \min_\beta \sum_{i\in \hat{\mathcal{I}}} \left( y_i - x_i^\prime \beta \right)^2 =
    \min_{
        \substack{
            \mathcal{S}_1 \subseteq \hat{\mathcal{I}}, \mathcal{S}_2 \subseteq \hat{\mathcal{O}} \\
            \left\vert \mathcal{S}_1 \right\vert \leq \left\vert \mathcal{S}_2 \right\vert \leq l
        }
    } \min_{\beta} \sum_{i\in \left( \hat{\mathcal{I}} \setminus \mathcal{S}_1 \right) \cup \mathcal{S}_2} \left( y_i - x_i^\prime \beta \right)^2.
\end{equation}
If solution $ \hat{\alpha} $  is swap-inescapable of order $ k $, then $ \hat{\alpha} $, associated with the OLS estimates $ \hat{\beta} $ using observations in $ \hat{\mathcal{I}} $, is the exact global optimal solution to \eqref{eq:L0_prob}. By solving the local combinatorial search problem in \eqref{eq:local_comb_search_prob} for small $ l < k $, say $ l = 1$ or $2$, we improve the IHT algorithm and verify the \textit{local combinatorial exactness}. 

The problem in \eqref{eq:local_comb_search_prob} can be formulated as a mixed-integer optimization (MIO) problem given by
\begin{align}
    & \min_{\beta, \gamma, \alpha} L_N\left( \beta, \alpha \right) \nonumber \\
    \text{ s.t. } & 
    \gamma_i \in \left\{ 0, 1 \right\},\, i \in \left[ N \right], \nonumber\\
    & \sum_{i = 1 }^{N}\gamma_i \leq k, \nonumber\\
    & \left(\alpha_i, 1 - \gamma_i  \right) \in \text{SOS-1},\, i \in \left[ N \right],   \label{eq:local_comb_search_mio}\\
    & \sum_{i\in\hat{\mathcal{I}}} \gamma_i \leq l,\, \nonumber\\
    & \sum_{i\in\hat{\mathcal{O}}} \gamma_i \geq k - l, \nonumber\\
    & \left\lVert \alpha  \right\rVert_{\infty} \leq   M_{\alpha},\, \left\lVert \alpha \right\rVert_1 \leq M_{\alpha, 1}. \nonumber
\end{align}
The constraints \(\sum_{i\in\hat{\mathcal{I}}} \gamma_i \leq l\) and \(\sum_{i\in\hat{\mathcal{O}}} \gamma_i \geq k - l\) restrict the number of swaps between \(\hat{\mathcal{I}} \) and \(\hat{\mathcal{O}} \) up to \(l\). 
\eqref{eq:local_comb_search_mio} has a much smaller search space than the original problem \eqref{eq:mio_prob} when $ l $ is small.\footnote{
    When $ l = 1 $, we can solve the problem in a brute force way by running least square estimation on all possible 1-1 swaps between $ \hat{\mathcal{I}} $ and $ \hat{\mathcal{O}} $ and comparing the $ k\left( N - k \right) $ resulting objective values without invoking the MIO solver.
} The heuristics combining IHT and local combinatorial search are summarized in Algorithm \ref{algo:local_comb_search}. 

\begin{algorithm}[h]
    \caption{Local Combinatorial Search}
    \label{algo:local_comb_search}
    \SetKwInOut{Input}{input}
    \SetKwInOut{Output}{output}

    \Input{An initial robust estimator $\hat{\beta}^{(0)}$, sparsity parameter $ k $  and local exactness level $l$.}
    \For{j = 0, 1, 2, ...}{
        \begin{enumerate}
            \item {
                $ \left(\hat{\beta}^{(j+1)}, \hat{\alpha}^{(j+1)} \right) \gets $ Output of Algorithm \ref{algo:iht} (IHT) initialized with $ \hat{\beta}^{(j)} $; 
            }
            \item {
                Compute $ \hat{\mathcal{I}}^{(j+1)} = \left\{ i\in \mathcal{N}: \hat{\alpha}_i^{(j+1)} = 0 \right\}$ and $ \hat{\mathcal{O}}^{(j+1)} = \left\{ i\in \mathcal{N}: \hat{\alpha}_i^{(j+1)} \neq 0 \right\}$;
            }
            \item {
                $ \left(\tilde{\beta}, \tilde{\alpha} \right) \gets $ solutions to local combinatorial search problem in \eqref{eq:local_comb_search_mio} given $  \hat{\mathcal{I}}^{(j+1)} $, $  \hat{\mathcal{O}}^{(j+1)} $ and $ l $.
            }
        \end{enumerate}
        \eIf{
            $ L_N\left( 
                \tilde{\beta}, \tilde{\alpha}
             \right) < 
            L_N\left( 
                \hat{\beta}^{(j+1)}, \hat{\alpha}^{(j+1)}
             \right) $
        }{
          $ \left( \hat{\beta}^{(j+1)}, \hat{\alpha}^{(j+1)} \right) \gets  \left( 
            \tilde{\beta}, \tilde{\alpha}\right) $.
        }
        {
            \Output{$\hat{\beta} = \hat{\beta}^{(j+1)}$ and $\hat{\alpha}= \hat{\alpha}^{(j+1)}$.}
        }
    }
    
\end{algorithm}

\subsubsection{Neighborhood Search}
Algorithm \ref{algo:local_comb_search} is guaranteed to provide a solution that is locally optimal in the sense of swap-inescapability of order $ l $. However, the quality of the solution depends on the initial inputs due to the noncovexity. As noted in \citet{thompson2022robust} and \citet{mazumder/radchenko/dedieu:2022:subsetselection}, a neighborhood search procedure can serve as a useful systematic way of perturbing the initial inputs to improve the solution quality of the heuristic algorithm. Let $ [K] = \{1, 2, \cdots, K\}$ with $ K \leq \lfloor N/2\rfloor $ be the set of candidate sparsity parameters. 

Denote $ b_{k, l}\left( \beta \right) $ and $ a_{k, l}\left( \beta \right) $ as the outputs of Algorithm \ref{algo:local_comb_search} initialized with $ \beta $,  $ k $ and $ l $. The neighborhood search procedure is summarized in Algorithm \ref{algo:neighborhood_search}. As a by-product of the procedure, we obtain a solution for each $ k\in \left[ K \right]  $ which is useful for sensitivity analysis and parameter tuning, which is discussed in Section \ref{sec:bic}.

\begin{algorithm}
    \caption{Neighborhood Search}
    \label{algo:neighborhood_search}
    \SetKwInOut{Input}{input}
    \SetKwInOut{Output}{output}

    \Input{An initial robust estimator $\hat{\beta}^{(0)}$, $ \left[ K \right]  = \left\{ 1, 2, \cdots, K \right\}$  and local exactness level $l$.  }

    \For{k = 1, 2, ..., K}{
        $\hat{\beta}^{(1)}_k \gets b_{k,l}\left( \hat{\beta}^{(0)} \right)$; $\hat{\alpha}^{(1)}_k \gets a_{k,l}\left( \hat{\beta}^{(0)} \right)$.
    }

    Initialize $ j = 1 $, \While{
        $ \left\vert 
            \sum_{k=1}^{K} L_N\left( \beta^{(j)}, \alpha^{(j)} \right) - \sum_{k=1}^{K} L_N\left( \beta^{(j-1)}, \alpha^{(j-1)} \right)
         \right\vert > 0$ 
    }{
        \For{k = 1, 2, ..., K}{
            $\tilde{\beta}_{-} \gets b_{k,l}\left( \hat{\beta}^{(j)}_{k-1} \right)$; $\tilde{\alpha}_{-} \gets a_{k,l}\left( \hat{\beta}^{(j)}_{k-1} \right)$;
            $\tilde{\beta}_{+} \gets b_{k,l}\left( \hat{\beta}^{(j)}_{k+1} \right)$; $\tilde{\alpha}_{+} \gets a_{k,l}\left( \hat{\beta}^{(j)}_{k+1} \right)$;
            $ \left( \tilde{\beta}, \tilde{\alpha} \right) \gets 
            \argmin\limits_{\left( \beta, \alpha \right) \in \left\{ 
                \left( \hat{\beta}^{(j)}_{k}, \hat{\alpha}^{(j)}_{k} \right),
                \left( \tilde{\beta}_{-}, \tilde{\alpha}_{-} \right),
                \left( \tilde{\beta}_{+}, \tilde{\alpha}_{+} \right)
            \right\}} L_N\left( \beta, \alpha \right)$;
            $ \hat{\beta}^{(j+1)}_{k} \gets \tilde{\beta}$; $ \hat{\alpha}^{(j+1)}_{k} \gets \tilde{\alpha}$.
        }
        $j \gets j + 1$
    }
    \Output{
        $\hat{\beta}_k \gets \hat{\beta}^{(j)}_{k}$, $\hat{\alpha}_k \gets \hat{\alpha}^{(j)}_{k}$, for $ k \in \left[ K \right]  $.
    }
\end{algorithm}

\subsection{Tuning Parameter Choice\label{sec:bic}}

In both \(L_1\)- and \(L_0\)- regularized estimation methods, the selection of tuning parameters, \(\psi \) in \eqref{eq:l1_prob} and \(k\) in \eqref{eq:L0_prob}, plays a critical role.
We propose 
the BIC-type information criteria, \begin{equation}
    \text{BIC}^\ast \left( k \right) = N\log \left( \frac{\left\Vert Y - 
    X\hat{\beta} - \hat{\alpha} \right\Vert_2^2}{N} \right) + k\log\left( N 
    \right). \label{eq:bic}
\end{equation}
For computational efficiency, we use Algorithm \ref{algo:neighborhood_search} 
to generate solutions for a grid of candidate tuning parameters and select the 
one that minimizes the BIC,
\begin{equation}
    \hat{k} =\argmin _{k\in \{1, 2,\cdots, K\}} \text{BIC}^\ast \left( k 
    \right),
    \label{eq:khat_bic}
\end{equation}
where $K \leq \lfloor N/2\rfloor $ is the maximum potential number of outliers.

The widely used theoretical property to evaluate robust estimation methods is the finite sample breakdown point, proposed in \citet{hampel1971general} and \citet{huber1983notion}. Suppose the original sample is $\left( Y, X \right)$ and the contaminated sample is $\left( \tilde{Y}_{(k_0)}, \tilde{X}_{(k_0)} \right)$ with $k_0$ observations in the original sample being arbitrarily replaced by outliers.  The finite sample breakdown point of an estimator $T$ is defined as
\begin{equation*}
    B\left( T; \left( Y, X \right) \right) = \min_{k_0}\left\{ 
        \frac{k_0}{N} \vert 
        \sup_{\tilde{Y}_{(k_0)}, \tilde{X}_{(k_0)}} \left\Vert 
            T\left( Y, X \right) - T\left( \tilde{Y}_{(k_0)}, \tilde{X}_{(k_0)} \right) 
         \right\Vert= \infty 
     \right\}.
\end{equation*}

\begin{proposition}
    \label{prop:breakdown_point}
    Let $\left( Y, X \right)$ be a sample of size $N$. $T\left( Y, X \right)$ denotes the estimator for $\beta$ obtained by solving \eqref{eq:L0_prob} with $\hat{k}$ selected by \eqref{eq:khat_bic} based on $\left( Y, X \right)$. For a fixed $k$, the estimator defined by \eqref{eq:L0_prob} with $k$ is denoted as $T_k\left( Y, X \right)$. Then, $T\left( X, Y \right)$ has the finite sample break down point
    $$B\left( T; \left( Y, X \right) \right) = \frac{\lfloor N/2\rfloor + 1}{N}.$$
\end{proposition}

\begin{proof}
The proof is relegated in Appendix \ref{app:proof}.
\end{proof}

\begin{remark}
    For a fixed $k$, $B\left( T_k; \left( Y, X \right) \right) = \frac{k + 1}{N}$, as shown in \citet{thompson2022robust}. Proposition \ref{prop:breakdown_point} extends the results to the estimation procedure with $\mathrm{BIC}^\ast\left( k \right)$ which achieves the optimal breakdown point.
\end{remark}

\section{Monte Carlo Simulation\label{sec:mc}}
In this section, we examine the numerical performance of the proposed $L_0$-regularized estimation procedure. We compare its coefficient estimation accuracy and prediction error to those of the \(L_1\)-regularized estimation and the classical methods, LAD and OLS.

\subsection{Setup}
Following the setting in Section \ref{sec:model}, we consider the following data generating processes (DGPs).

\medskip
\noindent \textbf{DGP 1} (Exogenous Outliers).
Consider the linear regression model with outliers,
\begin{equation*}
    y_i = \beta_0 + \beta_1 x_{i,1} + \beta_2 x_{i,2} + \gamma_i \alpha_i
    + u_i,\, i = 1, 2, \cdots, N,
\end{equation*}
where $\gamma_i$ is the indicator for outliers and $ \alpha_i $ is the outlier fixed effect. 
We generate the regressors by $x_{i, 1} = \left(v_{i,1}^2 + v_{i,2}^2 - 2\right)/2$ and $ x_{i, 2} = x_{i, 1} + v_{i, 3} $, where $v_{i, 
j} \sim \text{i.i.d.}N(0, 1)$ for $ j = 1, 2, 3 $, and the error term by $ u_i  \sim \text{i.i.d.}N(0, 1)$. 
Let $p\in \left( 0, 1 \right)$ denote the fraction of outliers. $k_0 = \left\lfloor pN \right\rfloor$ and 
$
    \gamma_i = \begin{cases}
        1 & i <= k_0\\
        0 & i > k_0
    \end{cases}.
$ 
Let $\alpha_i \sim \text{i.i.d.}N(\mu_\alpha, \sigma_\alpha^2)$ be exogenous shocks to outliers where $\left( \mu_\alpha,\sigma_\alpha \right) \in \left\{
    (0, 5), (5, 5), (10, 10)
\right\}$. The true coefficients are $ \beta_0 = 0.5 $,  $ \beta_1 = \left( 1, 1 \right)^\prime $.

\medskip
\noindent\textbf{DGP 2} (Endogenous Outliers). This DGP deviates from DGP 1 by allowing the outlier fixed effects to be correlated with the regressors. Let $ \alpha_i = \rho \left( v_{i, 1} + v_{i, 2}+ v_{i, 3}  \right) $ be a linear combination of the innovations to make the outlier fixed effects correlated with regressors and create endogeneity. The parameter $ \rho \in \left\{  2, 5, 10 \right\} $ to control the degree of correlation. The rest components are the same as DGP 1.

\medskip
\noindent{\textbf{DGP 3}} (Predictive Regression with Endogenous Shocks). 
Consider a linear predictive regression model as studied in \citet{kostakis2015robust,koo2020high,lee2022lasso}.  
The dependent variable is generated as
\begin{equation*}
y_{i+1}= \beta _{0} + \beta _{1} z_{i} +
\sum_{l=1}^{2}x_{i,l}^{c}\phi_{l} + \sum_{l=1}^{2}x_{i,l}\eta_{l,N} + \gamma _{i} \alpha _{i} + u_{i+1},
\end{equation*}
where $\beta _{0} = 0.3$, $\beta _{1} = 1$, \(\phi = \left( 1, -1 \right) \) and \(\eta _{N} = \left( 1 / \sqrt{N}, -1 / \sqrt{N} \right) \).
The vector of the stacked innovation $\xi_{i}=\left(z_{i},\underset{2\times1}{v_{i}^{\prime}},\underset{2\times1}{e_{i}^{\prime}},u_{i}\right)^{\prime}$
follows a VAR(1) process $\xi_{i}=\Phi\xi_{i-1}+\varepsilon_{i}$,
where $\varepsilon_{i}\sim iid\;N\left(0,\Sigma_{\varepsilon}\right)$ in
which $\Phi$ and $\Sigma_{\varepsilon}$ are empirically estimated
from the \citet{welch2008comprehensive} data as in \citet[Supplements S1]{lee2022lasso}.
$x_{i}^{c}\in\mathbb{R}^{2}$ is a vector I(1) process with
cointegration rank 1 based on the VECM, $\Delta x_{i}^{c}=\Gamma^{\prime}\Lambda x_{i-1}^{c}+v_{i},$
where $\Lambda=\begin{pmatrix}1 & -1\\
0 & 1
\end{pmatrix}$ and $\Gamma=\begin{pmatrix}0 & 1\\
1 & 0
\end{pmatrix}$ are the cointegrating matrix and the loading matrix, respectively.
$\left(x_{i,l}\right)_{l=1}^{2}$ are random walks generated by $x_{i,l}=x_{i-1,l}+e_{i,l}$, $l=1,2$. 
Let \(p \in (0,1)\) be the fraction of outliers and define \(k_0 = \lfloor pN \rfloor\). We impose two outlier periods. Let \(c_1 = \lfloor 0.25 N \rfloor\) and \(c_2 = \lfloor 0.75  N \rfloor\) be the positions in the sample where the outliers start. The outlier indicator \(\gamma_i = 1\) for \(i = c_l + 1, c_l + 2, \cdots, c_l + \lfloor k_0 / 2 \rfloor\), \(l = 1, 2\) and 0 otherwise. The outlier shifts \(\alpha _{i} = \rho \left( z_i + v_{i} ^{\prime} \iota_2 \right) \) where \(\rho  \in \left\{ 2,5,10 \right\} \).

\medskip

To evaluate each heuristic algorithm in Section \ref{subsec:heuristics} as compared to solving the full-scale mixed integer optimization for \eqref{eq:mio_prob}, we will report the measure \textit{Equal MIO}, which is the frequency among replications the estimates obtained by the algorithm are the same as those of the MIO solution of \eqref{eq:mio_prob}. For each implementation invoking the MIO solver, we report \textit{relative optimality gap}, which is defined as
\[
  \text{Relative Optimality Gap} = \frac{f^{P} - f^{D}}{f^{D}}, 
\]
where \(f^{P}\) is the attained upper bound of the objective value of the MIO solution and \(f^{D}\) is the dual lower bound of the objective value delivered by the solver.\footnote{The value can be read from the solver output. Details of the definition can be found at \href{https://www.gurobi.com/documentation/current/refman/mipgap2.html}{https://www.gurobi.com/documentation/current/refman/mipgap2.html}.} Relative optimality gap measures the progress of optimality verification of the solver. The solution is verified to be globally optimal if this gap is less than a prespecified tolerance threshold \(10^{-4}\). 

For the estimation and out-of-sample prediction performance, we report the bias, root mean squared error (RMSE) and prediction errors. Generically, bias and RMSE are calculated by 
$R^{-1}\sum_{r=1}^R \left( \hat{ \beta}^{(r)} - \beta \right)$ and 
$\sqrt{%
R^{-1}\sum_{r= 1}^R \left( \hat{\beta}^{(r)} -\beta \right)^2}$, respectively, for true parameter $\beta$, its estimate $\hat{\beta}%
^{(r)}$ across $R$ replications. To calculate the prediction error, we generate test data following the same DGP without outliers, \(\left\{ \widetilde{y} _{i}, \widetilde{x}^{\prime}_{i}  \right\}_{i=1}^{N_t} \), \(N_{t} = 1000\). The prediction error is calculated by 
\[
  \text{Prediction Error} = R^{-1} \sum_{r=1}^R \left( \frac{1}{N_t} \sum_{i=1}^{N_t} \left( \widetilde{y} _{i}^{(r)} - \widetilde{x} _{i}^{(r)\prime} \hat{\beta}^{(r)} \right)^2 \right).
\]

In all experiments, we consider $ p\in \left\{ 5\%, 10\%, 20\% \right\} $ and $ N \in \left\{ 100, 200, 400 \right\} $. All experiments are carried out on a Linux machine with an Intel i9-13900K CPU and \texttt{Guorbi} optimization solver at version 11.0. The time limit for the solver is set to 5 minutes.

\subsection{Performance of the Heuristic Algorithms}

We begin by evaluating the performance of the heuristic algorithms presented in Section \ref{subsec:heuristics} in comparison to solving the full-scale mixed integer optimization problem \eqref{eq:mio_prob}. 
We set the parameters \(\left( \mu _\alpha , \sigma _\alpha  \right) = \left( 5,5 \right)  \) for DGP 1 and \(\rho = 5 \) for DGP 2 and 3 and run \(100\) replications for each setup. The parameter \(k\) is set to the true value \(k_0\) for all algorithms. For MIO problems \eqref{eq:mio_prob} and \eqref{eq:local_comb_search_mio}, \(M_\alpha = \tau \left\lVert \hat{\alpha}   \right\rVert_\infty  \) and \(M_{\alpha ,1} = \tau \left\lVert \hat{\alpha}  \right\rVert_1 \)  where \(\hat{\alpha} \) is the initial estimator and \(\tau  = 1.5\).   

For each simulated dataset, we implement Algorithm \ref{algo:iht}, the iterative hard-thresholding (IHT) method. Subsequently, using the IHT estimate as the initial estimator, we apply the local combinatorial search (LCS) algorithm with local exactness levels \( l = 1\) and \( l = 2\), referred to as LCS-1 and LCS-2, respectively. Additionally, we solve the full-scale mixed integer optimization problem \eqref{eq:mio_prob} using the IHT estimate as the warm-start. 
In Table \ref{tab:algo_compare}, we report the average CPU runtime (in seconds) for each algorithm, the frequency of IHT, LCS-1 and LCS-2 estimates are equal to MIO and the average relative optimality gap. 

\begin{sidewaystable}[htbp]
    \begin{center}
        \caption{Heuristic Algorithms Performance}
        \label{tab:algo_compare}
        \begin{tabular}{ccc|cccc|ccc|ccc}
    \hline
     &  &  & \multicolumn{4}{c|}{CPU Time (Seconds)} & \multicolumn{3}{c|}{Equal MIO} & \multicolumn{3}{c}{Optimality Gap} \\ \cline{4-13} 
     & $N$ & $p$ & IHT & LCS-1 & LCS-2 & MIO & IHT & LCS-1 & LCS-2 & LCS-1 & LCS-2 & MIO \\ \hline
    \multicolumn{1}{c|}{\multirow{9}{*}{\rotatebox{90}{DGP 1}}} & 100 & 0.05 & 0.0004 & 0.0969 & 0.2030 & 0.1596 & 0.91 & 1.00 & 1.00 & 0.0001 & 0.0000 & 0.0000 \\
    \multicolumn{1}{c|}{} & 100 & 0.1 & 0.0004 & 0.1169 & 0.8459 & 5.2721 & 0.83 & 1.00 & 1.00 & 0.0001 & 0.0002 & 0.0000 \\
    \multicolumn{1}{c|}{} & 100 & 0.2 & 0.0004 & 0.2433 & 1.4303 & 222.29 & 0.59 & 0.93 & 0.99 & 0.0000 & 0.0003 & 0.1298 \\
    \multicolumn{1}{c|}{} & 200 & 0.05 & 0.0005 & 0.3674 & 5.2707 & 65.573 & 0.95 & 0.98 & 1.00 & 0.0000 & 0.0004 & 0.0131 \\
    \multicolumn{1}{c|}{} & 200 & 0.1 & 0.0005 & 0.6265 & 12.093 & 304.92 & 0.78 & 0.94 & 1.00 & 0.0001 & 0.0009 & 0.3527 \\
    \multicolumn{1}{c|}{} & 200 & 0.2 & 0.0005 & 0.7830 & 19.359 & 306.12 & 0.72 & 0.87 & 1.00 & 0.0002 & 0.0009 & 0.7530 \\
    \multicolumn{1}{c|}{} & 400 & 0.05 & 0.0007 & 1.7538 & 105.28 & 302.69 & 0.92 & 1.00 & 1.00 & 0.0000 & 0.0010 & 0.4249 \\
    \multicolumn{1}{c|}{} & 400 & 0.1 & 0.0007 & 4.3731 & 205.89 & 302.87 & 0.77 & 0.92 & 0.98 & 0.0006 & 0.0010 & 0.7960 \\
    \multicolumn{1}{c|}{} & 400 & 0.2 & 0.0008 & 4.7983 & 299.60 & 302.56 & 0.55 & 0.83 & 0.98 & 0.0008 & 0.0045 & 0.9500 \\ \hline
    \multicolumn{1}{c|}{\multirow{9}{*}{\rotatebox{90}{DGP 2}}} & 100 & 0.05 & 0.0004 & 0.1001 & 0.2034 & 0.2160 & 0.88 & 1.00 & 1.00 & 0.0001 & 0.0000 & 0.0000 \\
    \multicolumn{1}{c|}{} & 100 & 0.1 & 0.0004 & 0.1264 & 1.0286 & 2.4994 & 0.87 & 1.00 & 1.00 & 0.0002 & 0.0002 & 0.0000 \\
    \multicolumn{1}{c|}{} & 100 & 0.2 & 0.0004 & 0.1664 & 1.4692 & 76.835 & 0.65 & 0.94 & 0.99 & 0.0002 & 0.0004 & 0.0176 \\
    \multicolumn{1}{c|}{} & 200 & 0.05 & 0.0005 & 0.4235 & 6.0045 & 134.65 & 0.91 & 1.00 & 1.00 & 0.0001 & 0.0006 & 0.0402 \\
    \multicolumn{1}{c|}{} & 200 & 0.1 & 0.0005 & 0.6005 & 12.279 & 292.25 & 0.74 & 0.94 & 1.00 & 0.0001 & 0.0009 & 0.3163 \\
    \multicolumn{1}{c|}{} & 200 & 0.2 & 0.0005 & 0.8669 & 19.165 & 305.50 & 0.71 & 0.89 & 0.98 & 0.0003 & 0.0010 & 0.6237 \\
    \multicolumn{1}{c|}{} & 400 & 0.05 & 0.0006 & 1.9436 & 115.93 & 302.42 & 0.92 & 1.00 & 1.00 & 0.0002 & 0.0010 & 0.4884 \\
    \multicolumn{1}{c|}{} & 400 & 0.1 & 0.0007 & 4.9341 & 212.42 & 302.54 & 0.81 & 0.92 & 0.96 & 0.0007 & 0.0010 & 0.8046 \\
    \multicolumn{1}{c|}{} & 400 & 0.2 & 0.0007 & 5.2007 & 300.10 & 302.34 & 0.57 & 0.81 & 0.97 & 0.0008 & 0.0030 & 0.9354 \\ \hline
    \multicolumn{1}{c|}{\multirow{9}{*}{\rotatebox{90}{DGP 3}}} & 100 & 0.05 & 0.0008 & 0.3549 & 0.2814 & 0.3953 & 0.86 & 1.00 & 1.00 & 0.0001 & 0.0000 & 0.0000 \\
    \multicolumn{1}{c|}{} & 100 & 0.1 & 0.0009 & 0.4673 & 9.9887 & 13.779 & 0.86 & 0.99 & 1.00 & 0.0000 & 0.0005 & 0.0001 \\
    \multicolumn{1}{c|}{} & 100 & 0.2 & 0.0010 & 0.5272 & 13.365 & 214.53 & 0.60 & 0.96 & 0.99 & 0.0001 & 0.0007 & 0.2078 \\
    \multicolumn{1}{c|}{} & 200 & 0.05 & 0.0014 & 1.8395 & 38.356 & 135.95 & 0.79 & 0.97 & 1.00 & 0.0001 & 0.0008 & 0.0597 \\
    \multicolumn{1}{c|}{} & 200 & 0.1 & 0.0013 & 2.2910 & 64.206 & 293.49 & 0.75 & 0.95 & 0.99 & 0.0002 & 0.0011 & 0.4392 \\
    \multicolumn{1}{c|}{} & 200 & 0.2 & 0.0016 & 2.3277 & 61.629 & 302.97 & 0.63 & 0.84 & 0.95 & 0.0002 & 0.0009 & 0.8481 \\
    \multicolumn{1}{c|}{} & 400 & 0.05 & 0.0019 & 7.4183 & 241.30 & 301.51 & 0.87 & 0.96 & 1.00 & 0.0003 & 0.0062 & 0.5299 \\
    \multicolumn{1}{c|}{} & 400 & 0.1 & 0.0097 & 9.0912 & 265.04 & 301.38 & 0.80 & 0.88 & 0.95 & 0.0004 & 0.0075 & 0.8901 \\
    \multicolumn{1}{c|}{} & 400 & 0.2 & 0.0063 & 10.169 & 253.46 & 301.07 & 0.51 & 0.78 & 0.93 & 0.0004 & 0.0103 & 0.9900 \\ \hline
\end{tabular}
    \end{center}
    \par
    {\footnotesize \textit{Note:} The left panel reports the average computation time used for each algorithm in seconds. The middle panel reports the frequency among replications the estimates obtained by the algorithm are the same as those of the MIO solution of \eqref{eq:mio_prob}. The right panel corresponds to the average of the relative optimality gap defined as \(\frac{f^{P} - f^{D}}{f^{D}}\),
    where \(f^{P}\) is the attained upper bound of the objective value of the MIO solution and \(f^{D}\) is the dual lower bound of the objective value delivered by the solver, across replications. IHT refers to the iterative hard-thresholding in Algorithm \ref{algo:iht}. LCS-1 and LCS-2 are the local combinatorial search algorithm in Algorithm \ref{algo:local_comb_search} with \(l=1\) and \(2\), respectively. MIO refers to solving \eqref{eq:mio_prob} directly by \texttt{Gurobi} solver.}
\end{sidewaystable}

The key findings presented in Table 1 highlight the performance differences between the heuristic algorithms and the mixed integer optimization (MIO) approach. Notably, within the 5-minute time limit, the MIO method fails to complete the optimality verification for cases where the scale is larger than \(N = 200\) and \(p = 0.1\). In contrast, the iterative hard-thresholding (IHT) and the local combinatorial search (LCS-1) with \(l=1\) are significantly faster. Specifically, IHT runs in milliseconds and LCS-1 achieves optimality in seconds for most cases.
LCS-2 takes less than a minute except when \(N = 400\) and \(p = 0.1\) or \(0.2\). In these exceptional cases, the relative optimality gap of LCS-2 remains below 1\%, indicating near convergence. 
Furthermore, the initial estimator, IHT, provides solutions of the same quality to MIO in 50\% to 90\% of the cases. When advancing to LCS-1, the frequency of obtaining solutions equivalent to MIO increases to approximately 90\% while with a significantly lower computational cost. LCS-2 consistently matches the MIO solutions in nearly all replications. These findings underscore the effectiveness and computation efficiency of the heuristic algorithms. Particularly, the local combinatorial search algorithm balances the solution stability and computation costs, which is important when dealing with large sample sizes and many outliers.

\subsection{Comparison between \(L_0\) and \(L_1\)-Regularized Estimation}

In this section, we compare the performance of \(L_0\) and \(L_1\)-regularized estimation in different scenarios.  For the \(L_0\) and \(L_1\)-regularized estimation, the tuning parameters \(k\) and \(\psi\) are selected using the information criteria \eqref{eq:bic}. When choosing \(k\), we rely on the neighborhood search algorithm, as detailed in Algorithm \ref{algo:neighborhood_search} to generate the estimates for each candidate \(k \in \left[ K \right] \). Specifically, in the implementation of Algorithm \ref{algo:neighborhood_search}, we set \(K = 2k_0\) and the local exactness level \(l = 1\). 
We report LCS-2 with \(\hat{k}\) selected by BIC as the \(L_0\)-regularized estimation results. In addition, we include ordinary least square (OLS) and least absolute deviation (LAD) as benchmarks. Bias and root mean squared error (RMSE) are reported for the first slope parameter \(\beta_1\) in all data generating processes (DGPs). For each setup, we run 1000 replications.

The key findings from the Monte Carlo simulation reported in Table \ref{tab:est_compare_dgp1} to \ref{tab:est_compare_dgp3} reveal several important insights. In DGP 1, where the outlier fixed effects are exogenous and there is no endogeneity issue, the bias of the \(L_0\) method is slightly larger than that of the alternatives. However, the \(L_0\) method achieves the smallest RMSE and prediction error in most cases. Ordinary Least Squares (OLS) is not robust to outliers, and their RMSE and prediction error increase significantly as the magnitude of the outlier fixed effects becomes larger. The estimation accuracy of the \(L_1\) method is unstable when the sample size and the outlier fraction are small, but it outperforms Least Absolute Deviations (LAD) and OLS as \(N\) and \(p\) increase. In DGP 2 and 3, where endogenous outlier fixed effects are introduced, the \(L_1\)-regularized method, LAD, and OLS exhibit severe estimation bias. In contrast, the \(L_0\)-regularized method is free from estimation bias and achieves the smallest RMSE and prediction error. The accuracy gap widens as \(\rho\), the parameter controlling the degree of endogeneity, and the fraction of outliers increase. These findings demonstrate that the \(L_0\) method is more robust to the presence of endogenous outliers and provides more accurate estimation.

\begin{sidewaystable}[htbp]
    \begin{center}
    \caption{Bias, RMSE and Prediction Error Comparison: DGP 1}
    \label{tab:est_compare_dgp1}
    \begin{tabular}{ccc|rrrr|rrrr|rrrr}
  \hline
   &  &  & \multicolumn{4}{c|}{Bias} & \multicolumn{4}{c|}{RMSE} & \multicolumn{4}{c}{Prediction Error} \\ \cline{4-15} 
  $N$ & $p$ & $\left( \mu_\alpha,\sigma_\alpha \right)$ & \multicolumn{1}{c}{L0} & \multicolumn{1}{c}{L1} & \multicolumn{1}{c}{LAD} & \multicolumn{1}{c|}{OLS} & \multicolumn{1}{c}{L0} & \multicolumn{1}{c}{L1} & \multicolumn{1}{c}{LAD} & \multicolumn{1}{c|}{OLS} & \multicolumn{1}{c}{L0} & \multicolumn{1}{c}{L1} & \multicolumn{1}{c}{LAD} & \multicolumn{1}{c}{OLS} \\ \hline
  100 & 0.05 & $(0,5)$ & 0.0178 & -0.0059 & -0.0068 & 0.0044 & 0.1754 & 0.4331 & 0.2591 & 0.2172 & 1.0455 & 1.0754 & 1.0691 & 1.0708 \\
  100 & 0.05 & $(5,5)$ & 0.0163 & 0.0061 & -0.0013 & -0.0119 & 0.1687 & 0.2567 & 0.2372 & 0.2618 & 1.0471 & 1.0513 & 1.0647 & 1.1550 \\
  100 & 0.05 & $(5,10)$ & 0.0095 & 0.0174 & -0.0084 & -0.0001 & 0.1740 & 0.4375 & 0.2403 & 0.4806 & 1.0462 & 1.0899 & 1.0640 & 1.5507 \\ \hline
  100 & 0.1 & $(0,5)$ & 0.0195 & -0.0032 & -0.0009 & -0.0018 & 0.2057 & 0.3765 & 0.4099 & 0.2793 & 1.0634 & 1.0752 & 1.1189 & 1.1139 \\
  100 & 0.1 & $(5,5)$ & 0.0197 & -0.0129 & 0.0051 & 0.0071 & 0.1952 & 0.4613 & 0.2992 & 0.3372 & 1.0713 & 1.1072 & 1.1006 & 1.3960 \\
  100 & 0.1 & $(5,10)$ & 0.0217 & -0.0251 & -0.0053 & -0.0206 & 0.1907 & 0.5812 & 0.2687 & 0.6476 & 1.0726 & 1.1381 & 1.0992 & 2.5156 \\ \hline
  100 & 0.2 & $(0,5)$ & 0.0598 & 0.0013 & -0.0037 & -0.0014 & 0.2537 & 0.2037 & 0.3693 & 0.3467 & 1.1117 & 1.0707 & 1.1242 & 1.1861 \\
  100 & 0.2 & $(5,5)$ & 0.0657 & -0.0023 & 0.0048 & -0.0133 & 0.2930 & 0.4674 & 0.2858 & 0.4629 & 1.2384 & 1.2532 & 1.1423 & 2.2922 \\
  100 & 0.2 & $(5,10)$ & 0.0886 & 0.0044 & -0.0041 & -0.0359 & 0.3756 & 0.3443 & 0.3889 & 0.8796 & 1.2692 & 1.5191 & 1.1856 & 6.0357 \\ \hline
  200 & 0.05 & $(0,5)$ & 0.0055 & -0.0011 & -0.0005 & -0.0020 & 0.1276 & 0.1125 & 0.1802 & 0.1504 & 1.0227 & 1.0185 & 1.0320 & 1.0338 \\
  200 & 0.05 & $(5,5)$ & 0.0091 & 0.0069 & 0.0073 & -0.0026 & 0.1226 & 0.2324 & 0.1887 & 0.1904 & 1.0218 & 1.0310 & 1.0340 & 1.1070 \\
  200 & 0.05 & $(5,10)$ & 0.0074 & 0.0036 & 0.0107 & 0.0086 & 0.1159 & 0.1134 & 0.2214 & 0.3237 & 1.0249 & 1.0246 & 1.0448 & 1.3797 \\ \hline
  200 & 0.1 & $(0,5)$ & 0.0146 & 0.0089 & 0.0060 & 0.0030 & 0.1363 & 0.2345 & 0.1631 & 0.1898 & 1.0254 & 1.0298 & 1.0293 & 1.0527 \\
  200 & 0.1 & $(5,5)$ & 0.0166 & 0.0137 & 0.0047 & 0.0121 & 0.1278 & 0.2367 & 0.1525 & 0.2437 & 1.0385 & 1.0539 & 1.0449 & 1.3352 \\
  200 & 0.1 & $(5,10)$ & 0.0140 & 0.0053 & 0.0117 & -0.0005 & 0.1269 & 0.1355 & 0.1899 & 0.4503 & 1.0331 & 1.0522 & 1.0451 & 2.2763 \\ \hline
  200 & 0.2 & $(0,5)$ & 0.0382 & -0.0025 & 0.0039 & -0.0155 & 0.1833 & 0.1406 & 0.2778 & 0.2525 & 1.0544 & 1.0321 & 1.0520 & 1.0980 \\
  200 & 0.2 & $(5,5)$ & 0.0330 & 0.0049 & 0.0018 & 0.0001 & 0.1775 & 0.1669 & 0.2220 & 0.3259 & 1.0985 & 1.1675 & 1.0926 & 2.1352 \\
  200 & 0.2 & $(5,10)$ & 0.0337 & -0.0013 & 0.0064 & -0.0189 & 0.1907 & 0.2270 & 0.2291 & 0.6181 & 1.0957 & 1.3804 & 1.1007 & 5.4837 \\ \hline
  400 & 0.05 & $(0,5)$ & 0.0033 & 0.0071 & 0.0021 & 0.0039 & 0.0822 & 0.2162 & 0.0944 & 0.1065 & 1.0119 & 1.0207 & 1.0140 & 1.0194 \\
  400 & 0.05 & $(5,5)$ & 0.0055 & 0.0028 & 0.0026 & 0.0024 & 0.0825 & 0.0789 & 0.0962 & 0.1294 & 1.0113 & 1.0154 & 1.0148 & 1.0867 \\
  400 & 0.05 & $(5,10)$ & 0.0041 & 0.0031 & 0.0033 & 0.0089 & 0.0800 & 0.0807 & 0.1000 & 0.2409 & 1.0132 & 1.0173 & 1.0165 & 1.3312 \\ \hline
  400 & 0.1 & $(0,5)$ & 0.0012 & -0.0024 & -0.0013 & -0.0064 & 0.0872 & 0.0852 & 0.1042 & 0.1310 & 1.0122 & 1.0119 & 1.0153 & 1.0273 \\
  400 & 0.1 & $(5,5)$ & 0.0052 & -0.0007 & 0.0028 & -0.0057 & 0.0857 & 0.0912 & 0.1288 & 0.1710 & 1.0158 & 1.0433 & 1.0277 & 1.2877 \\
  400 & 0.1 & $(5,10)$ & 0.0069 & 0.0019 & 0.0073 & -0.0055 & 0.0833 & 0.0911 & 0.1721 & 0.3137 & 1.0131 & 1.0351 & 1.0299 & 2.1209 \\ \hline
  400 & 0.2 & $(0,5)$ & 0.0023 & -0.0003 & -0.0017 & 0.0019 & 0.1051 & 0.1033 & 0.1025 & 0.1779 & 1.0152 & 1.0153 & 1.0150 & 1.0465 \\
  400 & 0.2 & $(5,5)$ & 0.0134 & 0.0040 & 0.0032 & 0.0063 & 0.1099 & 0.1130 & 0.2314 & 0.2254 & 1.0487 & 1.1439 & 1.0823 & 2.0636 \\
  400 & 0.2 & $(5,10)$ & 0.0172 & -0.0030 & -0.0059 & -0.0163 & 0.1215 & 0.1501 & 0.1191 & 0.4302 & 1.0602 & 1.3277 & 1.0741 & 5.2611 \\ \hline
  \end{tabular} 
    \end{center}
    \par
    {\footnotesize \textit{Note:} L0, L1, LAD and OLS refer to local combinatorial search algorithm with \(l = 2\), \(L_1\)-regularized estimator defined in \eqref{eq:l1_prob}, the least absolute deviation estimator and the ordinary least squares estimator, respectively. Generically, bias and RMSE are calculated by 
    $R^{-1}\sum_{r=1}^R \left( \hat{ \beta}^{(r)} - \beta \right)$ and 
    $\sqrt{%
    R^{-1}\sum_{r= 1}^R \left( \hat{\beta}^{(r)} -\beta \right)^2}$, respectively, for true parameter $\beta$, its estimate $\hat{\beta}%
    ^{(r)}$ across $R = 1000$ replications. We report bias and RMSE for \(\beta _1\). 
    The prediction error is calculated by 
    \(
    \text{Prediction Error} = R^{-1} \sum_{r=1}^R \left( \frac{1}{N_t} \sum_{i=1}^{N_t} \left( \widetilde{y} _{i}^{(r)} - \widetilde{x} _{i}^{(r)\prime} \hat{\beta}^{(r)} \right)^2 \right),
    \) where \(\left\{ \widetilde{y} _{i}, \widetilde{x} ^{\prime} _{i}   \right\}_{i=1}^{N_t} \) is generated from the same DGP without outliers and \(N_{t}  = 1000\).  }
\end{sidewaystable}

\begin{sidewaystable}[htbp]
    \begin{center}
    \caption{Bias, RMSE and Prediction Error Comparison: DGP 2}
    \label{tab:est_compare_dgp2}
    \begin{tabular}{ccc|rrrr|rrrr|rrrr}
  \hline
   &  &  & \multicolumn{4}{c|}{Bias} & \multicolumn{4}{c|}{RMSE} & \multicolumn{4}{c}{Prediction Error} \\ \cline{4-15} 
  $N$ & $p$ & $\rho$ & \multicolumn{1}{c}{L0} & \multicolumn{1}{c}{L1} & \multicolumn{1}{c}{LAD} & \multicolumn{1}{c|}{OLS} & \multicolumn{1}{c}{L0} & \multicolumn{1}{c}{L1} & \multicolumn{1}{c}{LAD} & \multicolumn{1}{c|}{OLS} & \multicolumn{1}{c}{L0} & \multicolumn{1}{c}{L1} & \multicolumn{1}{c}{LAD} & \multicolumn{1}{c}{OLS} \\ \hline
  100 & 0.05 & 2 & 0.0136 & -0.0425 & -0.0213 & -0.1039 & 0.1856 & 0.4412 & 0.2999 & 0.2418 & 1.0520 & 1.1012 & 1.0764 & 1.0753 \\
  100 & 0.05 & 5 & 0.0108 & -0.0292 & -0.0293 & -0.2481 & 0.1721 & 0.3388 & 0.3421 & 0.4599 & 1.0454 & 1.0695 & 1.0981 & 1.2632 \\
  100 & 0.05 & 10 & 0.0024 & -0.0455 & -0.0375 & -0.5083 & 0.1695 & 0.5241 & 0.2712 & 0.9381 & 1.0450 & 1.1111 & 1.0725 & 2.0180 \\ \hline
  100 & 0.1 & 2 & 0.0124 & -0.0782 & -0.0647 & -0.2000 & 0.1985 & 0.2809 & 0.3104 & 0.3223 & 1.0584 & 1.0694 & 1.0793 & 1.1250 \\
  100 & 0.1 & 5 & 0.0177 & -0.1038 & -0.0633 & -0.5106 & 0.1959 & 0.4527 & 0.3108 & 0.7374 & 1.0617 & 1.0868 & 1.0951 & 1.6280 \\
  100 & 0.1 & 10 & 0.0084 & -0.1197 & -0.0665 & -0.9777 & 0.1932 & 0.3861 & 0.2889 & 1.4546 & 1.0652 & 1.1075 & 1.0973 & 3.4246 \\ \hline
  100 & 0.2 & 2 & 0.0470 & -0.1924 & -0.1443 & -0.3914 & 0.2518 & 0.5470 & 0.3400 & 0.5008 & 1.0995 & 1.1424 & 1.1261 & 1.2899 \\
  100 & 0.2 & 5 & 0.0210 & -0.2921 & -0.1563 & -1.0116 & 0.2536 & 0.4487 & 0.3351 & 1.2382 & 1.1436 & 1.1982 & 1.1369 & 2.6585 \\
  100 & 0.2 & 10 & -0.0150 & -0.4833 & -0.1652 & -1.9442 & 0.2628 & 0.6554 & 0.4466 & 2.4123 & 1.1598 & 1.5000 & 1.1400 & 7.3165 \\ \hline
  200 & 0.05 & 2 & -0.0039 & -0.0472 & -0.0373 & -0.1026 & 0.1263 & 0.1222 & 0.1418 & 0.1779 & 1.0203 & 1.0179 & 1.0231 & 1.0373 \\
  200 & 0.05 & 5 & 0.0019 & -0.0439 & -0.0357 & -0.2472 & 0.1228 & 0.2360 & 0.1610 & 0.3677 & 1.0200 & 1.0288 & 1.0287 & 1.1555 \\
  200 & 0.05 & 10 & 0.0028 & -0.0480 & -0.0289 & -0.5072 & 0.1209 & 0.1247 & 0.1675 & 0.7180 & 1.0243 & 1.0235 & 1.0350 & 1.6064 \\ \hline
  200 & 0.1 & 2 & -0.0058 & -0.0950 & -0.0619 & -0.2023 & 0.1371 & 0.2891 & 0.1853 & 0.2717 & 1.0284 & 1.0398 & 1.0369 & 1.0844 \\
  200 & 0.1 & 5 & 0.0066 & -0.0861 & -0.0538 & -0.4758 & 0.1344 & 0.2542 & 0.2812 & 0.6015 & 1.0284 & 1.0424 & 1.0614 & 1.4231 \\
  200 & 0.1 & 10 & 0.0024 & -0.1185 & -0.0633 & -1.0126 & 0.1321 & 0.1844 & 0.1817 & 1.2477 & 1.0283 & 1.0408 & 1.0515 & 2.6637 \\ \hline
  200 & 0.2 & 2 & 0.0132 & -0.1876 & -0.1343 & -0.3975 & 0.1743 & 0.2374 & 0.2047 & 0.4621 & 1.0486 & 1.0667 & 1.0520 & 1.2298 \\
  200 & 0.2 & 5 & 0.0077 & -0.2762 & -0.1411 & -1.0094 & 0.1801 & 0.3342 & 0.3001 & 1.1396 & 1.0629 & 1.1241 & 1.0809 & 2.3628 \\
  200 & 0.2 & 10 & -0.0130 & -0.4515 & -0.1441 & -2.0215 & 0.1803 & 0.5364 & 0.2658 & 2.2574 & 1.0698 & 1.3223 & 1.0785 & 6.3780 \\ \hline
  400 & 0.05 & 2 & -0.0055 & -0.0437 & -0.0278 & -0.1000 & 0.0815 & 0.2212 & 0.0907 & 0.1433 & 1.0099 & 1.0209 & 1.0117 & 1.0243 \\
  400 & 0.05 & 5 & 0.0010 & -0.0513 & -0.0274 & -0.2519 & 0.0826 & 0.0952 & 0.1322 & 0.3206 & 1.0115 & 1.0136 & 1.0192 & 1.1145 \\
  400 & 0.05 & 10 & -0.0005 & -0.0543 & -0.0338 & -0.4964 & 0.0829 & 0.0981 & 0.1257 & 0.6225 & 1.0113 & 1.0136 & 1.0172 & 1.4262 \\ \hline
  400 & 0.1 & 2 & -0.0112 & -0.0978 & -0.0607 & -0.1992 & 0.0874 & 0.1306 & 0.1100 & 0.2350 & 1.0123 & 1.0219 & 1.0171 & 1.0623 \\
  400 & 0.1 & 5 & -0.0001 & -0.1024 & -0.0568 & -0.4896 & 0.0860 & 0.1349 & 0.2014 & 0.5579 & 1.0116 & 1.0225 & 1.0339 & 1.3413 \\
  400 & 0.1 & 10 & -0.0017 & -0.1138 & -0.0648 & -0.9929 & 0.0854 & 0.1473 & 0.1400 & 1.1230 & 1.0112 & 1.0249 & 1.0237 & 2.3312 \\ \hline
  400 & 0.2 & 2 & -0.0307 & -0.2113 & -0.1373 & -0.4016 & 0.1022 & 0.2347 & 0.1827 & 0.4328 & 1.0176 & 1.0613 & 1.0381 & 1.1956 \\
  400 & 0.2 & 5 & -0.0047 & -0.2648 & -0.1350 & -1.0076 & 0.1064 & 0.2933 & 0.2762 & 1.0737 & 1.0207 & 1.0937 & 1.0692 & 2.1861 \\
  400 & 0.2 & 10 & -0.0052 & -0.4260 & -0.1478 & -2.0080 & 0.1020 & 0.4701 & 0.1815 & 2.1353 & 1.0214 & 1.2335 & 1.0410 & 5.6672 \\ \hline
\end{tabular} 
    \end{center}
    \par
    {\footnotesize \textit{Note:} L0, L1, LAD and OLS refer to local combinatorial search algorithm with \(l = 2\), \(L_1\)-regularized estimator defined in \eqref{eq:l1_prob}, the least absolute deviation estimator and the ordinary least squares estimator, respectively. Generically, bias and RMSE are calculated by 
    $R^{-1}\sum_{r=1}^R \left( \hat{ \beta}^{(r)} - \beta \right)$ and 
    $\sqrt{%
    R^{-1}\sum_{r= 1}^R \left( \hat{\beta}^{(r)} -\beta \right)^2}$, respectively, for true parameter $\beta$, its estimate $\hat{\beta}%
    ^{(r)}$ across $R = 1000$ replications. We report bias and RMSE for \(\beta _1\). 
    The prediction error is calculated by 
    \(
    \text{Prediction Error} = R^{-1} \sum_{r=1}^R \left( \frac{1}{N_t} \sum_{i=1}^{N_t} \left( \widetilde{y} _{i}^{(r)} - \widetilde{x} _{i}^{(r)\prime} \hat{\beta}^{(r)} \right)^2 \right),
    \) where \(\left\{ \widetilde{y} _{i}, \widetilde{x} ^{\prime} _{i}   \right\}_{i=1}^{N_t} \) is generated from the same DGP without outliers and \(N_{t}  = 1000\). }
\end{sidewaystable}

\begin{sidewaystable}[htbp]
    \begin{center}
    \caption{Bias, RMSE and Prediction Error Comparison: DGP 3}
    \label{tab:est_compare_dgp3}
    \begin{tabular}{ccc|rrrr|rrrr|rrrr}
  \hline
   &  &  & \multicolumn{4}{c|}{Bias} & \multicolumn{4}{c|}{RMSE} & \multicolumn{4}{c}{Prediction Error} \\ \cline{4-15} 
  $N$ & $p$ & $\rho$ & \multicolumn{1}{c}{L0} & \multicolumn{1}{c}{L1} & \multicolumn{1}{c}{LAD} & \multicolumn{1}{c|}{OLS} & \multicolumn{1}{c}{L0} & \multicolumn{1}{c}{L1} & \multicolumn{1}{c}{LAD} & \multicolumn{1}{c|}{OLS} & \multicolumn{1}{c}{L0} & \multicolumn{1}{c}{L1} & \multicolumn{1}{c}{LAD} & \multicolumn{1}{c}{OLS} \\ \hline
  100 & 0.05 & 2 & 0.0133 & 0.0374 & 0.0305 & 0.0947 & 0.1181 & 0.1216 & 0.1343 & 0.1707 & 1.0929 & 1.1016 & 1.1199 & 1.1796 \\
  100 & 0.05 & 5 & 0.0028 & 0.0372 & 0.0288 & 0.2316 & 0.1120 & 0.1214 & 0.1335 & 0.3450 & 1.1013 & 1.0985 & 1.1204 & 1.3718 \\
  100 & 0.05 & 10 & -0.0014 & 0.0400 & 0.0233 & 0.4303 & 0.1137 & 0.1283 & 0.1332 & 0.6133 & 1.1252 & 1.1627 & 1.1821 & 2.7509 \\ \hline
  100 & 0.1 & 2 & 0.0202 & 0.0912 & 0.0705 & 0.2313 & 0.1379 & 0.1571 & 0.1586 & 0.2903 & 1.2159 & 1.2227 & 1.2517 & 1.3784 \\
  100 & 0.1 & 5 & 0.0051 & 0.1177 & 0.0699 & 0.5802 & 0.1244 & 0.1799 & 0.1529 & 0.6761 & 1.1360 & 1.1734 & 1.1954 & 2.3094 \\
  100 & 0.1 & 10 & 0.0104 & 0.1829 & 0.0787 & 1.1178 & 0.1154 & 0.2515 & 0.1547 & 1.3130 & 1.1591 & 1.2737 & 1.1658 & 6.6681 \\ \hline
  100 & 0.2 & 2 & 0.0313 & 0.2201 & 0.1568 & 0.4539 & 0.1541 & 0.2676 & 0.2196 & 0.4995 & 1.1839 & 1.2791 & 1.2140 & 1.7171 \\
  100 & 0.2 & 5 & 0.0128 & 0.3870 & 0.1699 & 1.1336 & 0.1371 & 0.4517 & 0.2315 & 1.2259 & 1.1921 & 1.6585 & 1.2529 & 4.9567 \\
  100 & 0.2 & 10 & 0.0178 & 0.7165 & 0.1692 & 2.2973 & 0.1379 & 0.8299 & 0.2350 & 2.4777 & 1.1395 & 3.0481 & 1.2506 & 17.9627 \\ \hline
  200 & 0.05 & 2 & 0.0090 & 0.0459 & 0.0326 & 0.1128 & 0.0837 & 0.0911 & 0.0937 & 0.1493 & 1.1170 & 1.1022 & 1.1237 & 1.1233 \\
  200 & 0.05 & 5 & 0.0021 & 0.0481 & 0.0299 & 0.2854 & 0.0834 & 0.0992 & 0.0998 & 0.3453 & 1.0353 & 1.0283 & 1.0434 & 1.3253 \\
  200 & 0.05 & 10 & -0.0005 & 0.0523 & 0.0292 & 0.5746 & 0.0736 & 0.0962 & 0.0924 & 0.6796 & 1.0445 & 1.0538 & 1.0657 & 2.2574 \\ \hline
  200 & 0.1 & 2 & 0.0122 & 0.0927 & 0.0650 & 0.2258 & 0.0927 & 0.1263 & 0.1145 & 0.2563 & 1.0816 & 1.0692 & 1.0776 & 1.1800 \\
  200 & 0.1 & 5 & 0.0066 & 0.1106 & 0.0712 & 0.5705 & 0.0815 & 0.1417 & 0.1155 & 0.6230 & 1.1258 & 1.1599 & 1.1389 & 2.1070 \\
  200 & 0.1 & 10 & -0.0008 & 0.1536 & 0.0654 & 1.1162 & 0.0822 & 0.1931 & 0.1145 & 1.2159 & 1.0499 & 1.1218 & 1.0782 & 4.7269 \\ \hline
  200 & 0.2 & 2 & 0.0336 & 0.2183 & 0.1547 & 0.4549 & 0.1168 & 0.2405 & 0.1856 & 0.4767 & 1.0496 & 1.0919 & 1.0602 & 1.4486 \\
  200 & 0.2 & 5 & 0.0125 & 0.3576 & 0.1620 & 1.1350 & 0.0987 & 0.3921 & 0.1970 & 1.1828 & 1.0410 & 1.4231 & 1.1345 & 4.2417 \\
  200 & 0.2 & 10 & 0.0075 & 0.6543 & 0.1656 & 2.2731 & 0.0988 & 0.7136 & 0.2011 & 2.3656 & 1.0560 & 2.2161 & 1.1482 & 13.1995 \\ \hline
  400 & 0.05 & 2 & 0.0084 & 0.0527 & 0.0319 & 0.1139 & 0.0575 & 0.0760 & 0.0680 & 0.1335 & 1.0229 & 1.0125 & 1.0149 & 1.0271 \\
  400 & 0.05 & 5 & 0.0033 & 0.0571 & 0.0337 & 0.2849 & 0.0544 & 0.0800 & 0.0712 & 0.3138 & 1.0249 & 1.0412 & 1.0431 & 1.3004 \\
  400 & 0.05 & 10 & 0.0023 & 0.0568 & 0.0334 & 0.5652 & 0.0551 & 0.0812 & 0.0722 & 0.6180 & 1.0471 & 1.0697 & 1.0608 & 2.2000 \\ \hline
  400 & 0.1 & 2 & 0.0136 & 0.1077 & 0.0675 & 0.2276 & 0.0660 & 0.1241 & 0.0949 & 0.2433 & 1.0291 & 1.0416 & 1.0408 & 1.1401 \\
  400 & 0.1 & 5 & 0.0062 & 0.1134 & 0.0697 & 0.5633 & 0.0597 & 0.1299 & 0.0945 & 0.5890 & 1.0542 & 1.0938 & 1.0799 & 1.9397 \\
  400 & 0.1 & 10 & 0.0011 & 0.1469 & 0.0691 & 1.1311 & 0.0583 & 0.1676 & 0.0971 & 1.1816 & 1.1474 & 1.1870 & 1.1448 & 4.3749 \\ \hline
  400 & 0.2 & 2 & 0.0243 & 0.2288 & 0.1505 & 0.4515 & 0.0795 & 0.2405 & 0.1671 & 0.4632 & 1.0798 & 1.1893 & 1.1186 & 1.5341 \\
  400 & 0.2 & 5 & 0.0057 & 0.3392 & 0.1541 & 1.1227 & 0.0658 & 0.3560 & 0.1717 & 1.1470 & 0.9963 & 1.2354 & 1.0458 & 3.7024 \\
  400 & 0.2 & 10 & 0.0053 & 0.6241 & 0.1604 & 2.2530 & 0.0656 & 0.6515 & 0.1781 & 2.2978 & 1.0662 & 1.9160 & 1.1003 & 11.8648 \\ \hline
  \end{tabular} 
    \end{center}
    \par
    {\footnotesize \textit{Note:} L0, L1, LAD and OLS refer to local combinatorial search algorithm with \(l = 2\), \(L_1\)-regularized estimator defined in \eqref{eq:l1_prob}, the least absolute deviation estimator and the ordinary least squares estimator, respectively. Generically, bias and RMSE are calculated by 
    $R^{-1}\sum_{r=1}^R \left( \hat{ \beta}^{(r)} - \beta \right)$ and 
    $\sqrt{%
    R^{-1}\sum_{r= 1}^R \left( \hat{\beta}^{(r)} -\beta \right)^2}$, respectively, for true parameter $\beta$, its estimate $\hat{\beta}%
    ^{(r)}$ across $R = 1000$ replications. We report bias and RMSE for \(\beta _1\). 
    The prediction error is calculated by 
    \(
    \text{Prediction Error} = R^{-1} \sum_{r=1}^R \left( \frac{1}{N_t} \sum_{i=1}^{N_t} \left( \widetilde{y} _{i}^{(r)} - \widetilde{x} _{i}^{(r)\prime} \hat{\beta}^{(r)} \right)^2 \right),
    \) where \(\left\{ \widetilde{y} _{i}, \widetilde{x} ^{\prime} _{i}   \right\}_{i=1}^{N_t} \) is generated from the same DGP without outliers and \(N_{t}  = 1000\). }
\end{sidewaystable}

\section{Empirical Illustration: Stock Return Predictability\label{sec:empirical}}
 
Linear predictive regression models have been extensively studied for stock return forecasting. 
For instance, \citet{koo2020high,lee2022lasso} and others have applied linear predictive regression to the \citet{welch2008comprehensive} dataset to investigate stock return predictability. 
Financial time series data often exhibit instability, particularly during periods such as the financial crisis. Including these periods in estimation and forecasting can lead to varying results. As an illustration, we apply data-driven \(L_0\) and \(L_1\)-regularized robust estimation methods to the \citet{welch2008comprehensive} dataset\footnote{Retrieved from \url{http://www.hec.unil.ch/agoyal/}} and evaluate the out-of-sample stock return prediction performance.

We use monthly data from January 1990 to December 2023, which covers the dot-com bubble period, the 2007-09 financial crisis, and the Covid-19 pandemic. The dependent variable, excess return, is defined as the difference between the continuously compounded return on the S\&P 500 index and the three-month Treasury bill rate, computed by 
\[
\text{ExReturn}_{i} = \log\left(\text{index}_{i}/\text{index}_{i-1}\right) - \log\left(1 + \text{tbl}_{i}/12\right).
\]
Twelve financial and macroeconomic variables\footnote{
    The predictors are 
    \texttt{Dividend Price Ratio}, \texttt{Dividend Yield}, \texttt{Earning Price Ratio}, \texttt{Term Spread}, \texttt{Default Yield Spread}, \texttt{Default Return Spread}, \texttt{Book-to-Market Ratio}, \texttt{Treasury Bill Rates}, \texttt{Long-Term Return}, \texttt{Net Equity Expansion}, \texttt{Stock Variance}, \texttt{Inflation}. Table \ref{tab:AR(1)-Coefficients-of} in the appendix summarizes the detailed description of each variable.
}, denoted by \(x_{i}\), are included in the model as predictors. We apply the \(L_0\) and \(L_1\)-regularized methods to estimate the model,
\[
    \text{ExReturn}_{i+1} = \beta_0 + x_{i}^{\prime}\beta_1 + \alpha_{i} + u_{i+1},
\]
for \(i \in \left[ N \right]\), and construct one-month-ahead forecasts \(\hat{f}_{N+1} = \hat{\beta}_0 + \hat{x}_{N}^{\prime}\hat{\beta}_1\) recursively with a 10-year rolling window for each month from January 2000 to December 2023. As in Monte Carlo experiments, LAD and OLS are included as benchmarks.

Table \ref{tab:welch_goyal_estimation_results} compares the mean prediction squared error (MPSE) among different methods across various forecasting periods. We consider the forecast period from January 2000 to December 2023, and six subperiods based on the start and end dates of the dot-com bubble burst period (Mar. 2000 to Nov. 2000), the financial crisis (Dec. 2007 to Jun. 2009) and the Covid-19 shocks (Feb. 2020 to Apr. 2020). The results demonstrate that \(L_0\)-regularized method outperforms alternative methods with a margin in all subperiods except for the financial crisis period, which illuminates the forecast accuracy gain from the robustness to potentially endogenous outliers.

\begin{table}[htbp]
    \begin{center}
     \caption{Mean Prediction Squared Error (MPSE) for Monthly Excess Return of S\&P 500 Index}
        \label{tab:welch_goyal_estimation_results}
        \begin{tabular}{c|cccc}
            \hline
            Forecast Period & \multicolumn{1}{c}{L0} & \multicolumn{1}{c}{L1} & \multicolumn{1}{c}{LAD} & \multicolumn{1}{c}{OLS} \\ \hline
            All periods & \multirow{2}{*}{\textbf{0.00306}} & \multirow{2}{*}{0.00325} & \multirow{2}{*}{0.00400} & \multirow{2}{*}{0.00341} \\
            (\textit{\small Jan. 2000 to Dec. 2023}) &  &  &  &  \\ \hline
            Dot-com Bubble Burst & \multirow{2}{*}{\textbf{0.00409}} & \multirow{2}{*}{0.00455} & \multirow{2}{*}{0.00450} & \multirow{2}{*}{0.00474} \\
            (\textit{\small Mar. 2000 to Nov. 2000}) &  &  &  &  \\ \hline
            Dot-com - Financial Crisis & \multirow{2}{*}{\textbf{0.00169}} & \multirow{2}{*}{0.00187} & \multirow{2}{*}{0.00207} & \multirow{2}{*}{0.00197} \\
            (\textit{\small Mar. 2001 to Jun. 2009}) &  &  &  &  \\ \hline
            Financial Crisis & \multirow{2}{*}{0.00937} & \multirow{2}{*}{\textbf{0.00906}} & \multirow{2}{*}{0.01389} & \multirow{2}{*}{0.01082} \\
            (\textit{\small Dec. 2007 to Jun. 2009}) &  &  &  &  \\ \hline
            Financial Crisis - Covid & \multicolumn{1}{c}{\multirow{2}{*}{\textbf{0.00128}}} & \multicolumn{1}{c}{\multirow{2}{*}{0.00134}} & \multicolumn{1}{c}{\multirow{2}{*}{0.00148}} & \multicolumn{1}{c}{\multirow{2}{*}{0.00145}} \\
            (\textit{\small Dec. 2007 to Feb. 2020}) & \multicolumn{1}{c}{} & \multicolumn{1}{c}{} & \multicolumn{1}{c}{} & \multicolumn{1}{c}{} \\ \hline
            Covid & \multicolumn{1}{c}{\multirow{2}{*}{\textbf{0.06758}}} & \multicolumn{1}{c}{\multirow{2}{*}{0.07493}} & \multicolumn{1}{c}{\multirow{2}{*}{0.11073}} & \multicolumn{1}{c}{\multirow{2}{*}{0.07988}} \\
            (\textit{\small Feb. 2020 to Apr. 2020}) & \multicolumn{1}{c}{} & \multicolumn{1}{c}{} & \multicolumn{1}{c}{} & \multicolumn{1}{c}{} \\ \hline
            Post-Covid & \multicolumn{1}{c}{\multirow{2}{*}{\textbf{0.00332}}} & \multicolumn{1}{c}{\multirow{2}{*}{0.00378}} & \multicolumn{1}{c}{\multirow{2}{*}{0.00335}} & \multicolumn{1}{c}{\multirow{2}{*}{0.00340}} \\
            (\textit{\small Apr. 2020 to Dec. 2023}) & \multicolumn{1}{c}{} & \multicolumn{1}{c}{} & \multicolumn{1}{c}{} & \multicolumn{1}{c}{} \\ \hline
    \end{tabular}   
    \end{center}
    \par
    {\footnotesize \textit{Notes:} The mean prediction squared error (MPSE) is calculated by averaging the square forecasting loss over the corresponding periods. The date below each period refers to the forecast period. L0, L1, LAD and OLS refer to local combinatorial search algorithm with \(l = 2\), \(L_1\)-regularized estimator defined in \eqref{eq:l1_prob}, the least absolute deviation estimator and the ordinary least squares estimator, respectively. The smallest MPSE in each period is labeled in bold.} 
\end{table}

Additionally, Figure \ref{fig:combined_outlier_detection} illustrates the outliers detected (\(\hat{\alpha}_{i} \neq 0\)) for each rolling window using \(L_0\) and \(L_1\)-regularized methods. In the grid plot, each row corresponds to a different forecast period, and each column corresponds to a period that may be included in the estimation rolling window. Within each row, the highlighted cells represent the estimation window for the corresponding forecast target period. A cell is labeled red if the corresponding period is detected as an outlier in the rolling window by both methods, blue if both methods detect this period as an inlier, purple if only the \(L_0\)-regularized method detects this period as an outlier, and yellow if only the \(L_1\)-regularized method detects this period as an outlier. As shown in the figure, the \(L_0\) method consistently detects periods around the dotcom bubble, financial crisis, and Covid-19 shocks across rolling windows. In general, the \(L_1\) method detects a similar pattern while also labeling some outlier periods outside the concentrated outlier regions.

\begin{sidewaysfigure}[htbp]
    \centering
    \caption{Outlier Detection with \(L_0\) and \(L_1\)-regularized Methods}
    \label{fig:combined_outlier_detection}
    \includegraphics[width=\textwidth]{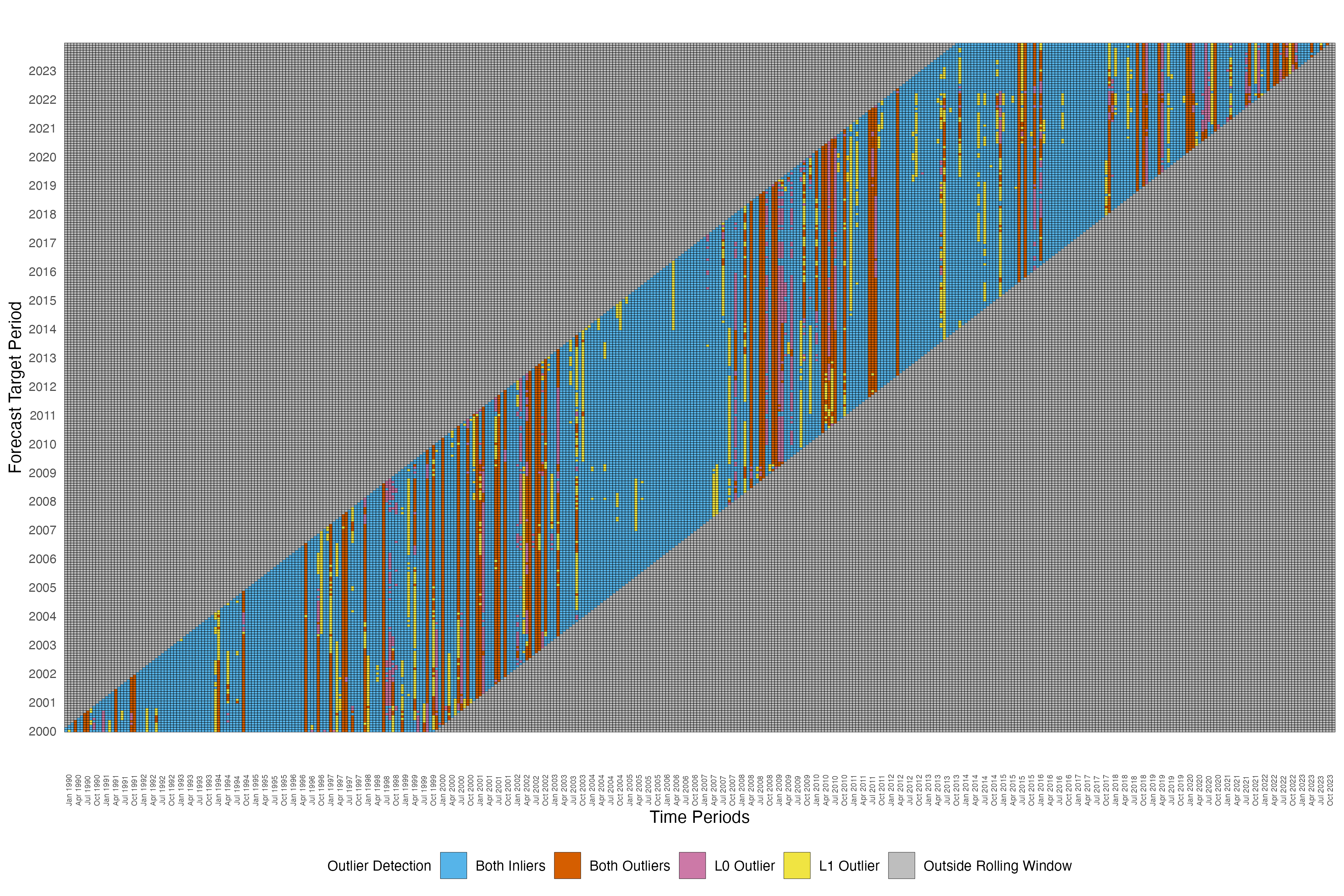}
\end{sidewaysfigure}

\section{Conclusion\label{sec:conclusion}}
This paper addresses the robust estimation of linear regression models in the presence of potentially endogenous outliers. We demonstrate that existing $L_1$-regularized estimation methods exhibit significant bias when outliers are endogenous and develop $L_0$-regularized estimation methods to overcome this issue. We propose systematic heuristic algorithms, notably an iterative hard-thresholding algorithm and a local combinatorial search refinement, to efficiently solve the combinatorial optimization problem of the \(L_0\)-regularized estimation. The properties of the \(L_0\) and \(L_1\)-regularized methods are examined through Monte Carlo simulations. We illustrate the practical value of our method with an empirical application to stock return forecasting.

\newpage
\appendix

\begin{center}
{\huge \bf Appendix}
\end{center}

\setcounter{table}{0} \renewcommand{\thetable}{A.\arabic{table}} 
\setcounter{section}{0} \renewcommand{\thesection}{A.\arabic{section}} 
\setcounter{figure}{0} \renewcommand{\thefigure}{A.\arabic{figure}}
\setcounter{remark}{0} \renewcommand{\theremark}{A.\arabic{remark}}
\setcounter{lemma}{0}
\renewcommand{\thelemma}{A.\arabic{lemma}}
\setcounter{equation}{0}
\renewcommand{\theequation}{A.\arabic{equation}}


\section{Proofs\label{app:proof}}

\begin{proof}[Proof of Proposition \ref{prop:breakdown_point}]
    Note that \eqref{eq:L0_prob} is equivalent to
    \begin{equation*}
        \Theta_k \left( Y, X \right) = \min_{
            \mathcal{I}\subset \mathcal{N}, 
            \left\vert \mathcal{I} \right\vert \geq  N - k
        } \min_{\beta} 
        \sum_{i\in \mathcal{I}} \left( y_i - x_i^\prime \beta \right)^2.
    \end{equation*}
    Denote $$\Theta \left( Y, X \right) = \Theta_{\hat{k}}\left( Y, X \right)$$ where 
    $$
    \hat{k}  = \arg\min_{1\leq k\leq\lfloor N/2\rfloor} 
        \mathrm{BIC}^\ast\left( k; \left( Y, X \right) \right) 
        = \arg\min_{1\leq k\leq\lfloor N/2\rfloor} 
        N \log \left( \frac{\Theta_k\left( Y, X \right)}{N} \right)
        + k\log\left( N \right).
    $$
    Note that $B\left( T; \left( Y, X \right) \right) = B\left( \Theta; \left( Y, X \right) \right)$ since $\Theta\left( Y, X \right)$ is a quadratic function of $T\left( Y, X \right)$, so we can focus on $\Theta\left( Y, X \right)$ as proceed.

    We first show $B\left( \Theta; \left( Y, X \right) \right) > \lfloor N / 2\rfloor$ by contradiction. Suppose $k_0 = \lfloor N / 2\rfloor$, and the set of uncontaminated sample is $\mathcal{I}_0$ with $\left\vert \mathcal{I}_0 \right\vert = N - \lfloor N / 2\rfloor$. The objective value with sample in $\mathcal{I}_0$ satisfies
    \begin{equation*}
        \Theta_{\lfloor N/2 \rfloor} \left( \tilde{Y}_{(k_0)}, \tilde{X}_{(k_0)}\right) \leq 
        \min_{\beta} \sum_{i\in \mathcal{I}_0} \left( y_i - x_i^\prime \beta \right)^2 < \infty,
    \end{equation*}
    then
    $
        \sup_{\tilde{Y}_{(k_0)}, \tilde{X}_{(k_0)}} \Theta_{\lfloor N/2 \rfloor} \left(\tilde{Y}_{(k_0)}, \tilde{X}_{(k_0)} \right) < \infty,
    $
    and 
    \begin{equation}
        \sup_{\tilde{Y}_{(k_0)}, \tilde{X}_{(k_0)}} \mathrm{BIC}^\ast\left( \lfloor N/2\rfloor; \left( \tilde{Y}_{(k_0)}, \tilde{X}_{(k_0)} \right)\right) < \infty.
        \label{eq:finite_bic}
    \end{equation}
    Suppose
    $$
        \sup_{\tilde{Y}_{(k_0)}, \tilde{X}_{(k_0)}} \left\Vert 
            \Theta \left( Y, X \right) - \Theta \left( \tilde{Y}_{(k_0)}, \tilde{X}_{(k_0)} \right) 
        \right\Vert= \infty,
    $$
    then $$\sup_{\tilde{Y}_{(k_0)}, \tilde{X}_{(k_0)}} \Theta \left( \tilde{Y}_{(k_0)}, \tilde{X}_{(k_0)} \right) = \sup_{\tilde{Y}_{(k_0)}, \tilde{X}_{(k_0)}} \Theta_{\hat{k}} \left( \tilde{Y}_{(k_0)}, \tilde{X}_{(k_0)} \right) = \infty.$$
    This implies \begin{equation*}
        \sup_{\tilde{Y}_{(k_0)}, \tilde{X}_{(k_0)}} \mathrm{BIC}^\ast \left( k, \left( \tilde{Y}_{(k_0)}, \tilde{X}_{(k_0)} \right) \right) = \infty,
    \end{equation*}
    for $1\leq k\leq \lfloor N/2 \rfloor$, which contradicts to \eqref{eq:finite_bic}. As a result, $B\left( \Theta; \left( Y, X \right) \right) > \lfloor N / 2\rfloor / N$.

    Suppose $k_0 = \lfloor N / 2\rfloor + 1$, then $\Theta\left(\tilde{Y}_{(k_0)}, \tilde{X}_{(k_0)} \right) = \sum_{i\in \mathcal{I}^\ast} \left( y_i - x_i^\prime \beta \right)^2$ for some $\mathcal{I}^\ast \subset \mathcal{N}$ with at least one arbitrary outlier, which leads to $ \sup_{\tilde{Y}_{(k_0)}, \tilde{X}_{(k_0)}} \Theta\left(\tilde{Y}_{(k_0)}, \tilde{X}_{(k_0)} \right) = \infty$, and then $$
    \sup_{\tilde{Y}_{(k_0)}, \tilde{X}_{(k_0)}} \left\Vert 
            \Theta \left( Y, X \right) - \Theta \left( \tilde{Y}_{(k_0)}, \tilde{X}_{(k_0)} \right) 
        \right\Vert= \infty.
    $$ Then $B\left( \Theta; \left( Y, X \right) \right) \leq\frac{ \lfloor N / 2\rfloor + 1}{N}$, which completes the proof.
\end{proof}

\section{Data Description\label{app:data_description}}

Table \ref{tab:AR(1)-Coefficients-of} summarizes all the variables
included and the first-order autocorrelation coefficient of each variable
estimated for the whole sample period. As shown in the table, the
excess return has an estimated first-order autocorrelation coefficient
of $0.0444$, which indicates little persistence, similar to the default
return spread (dfr), the long-term return of government bonds (ltr),
stock variance (svar) and inflation (infl). The other predictors show
high persistence, with AR(1) coefficients greater than 0.95.

\begin{table}[htbp]
\caption{\label{tab:AR(1)-Coefficients-of}Variables and AR(1) Coefficients}

\medskip{}

\centering
\small
\begin{tabular}{l>{\raggedright}p{11cm}r}
\toprule 
Predictor & Description & AR(1) Coef\tabularnewline
\midrule
ExReturn & {\small{}Excess Return: the difference between the continuously compounded
return on the S\&P 500 index and the three-month Treasury Bill rate} & 0.0444 \tabularnewline
dp & {\small{}Dividend Price Ratio: the difference between the log of the
12-month moving sum dividends and the log of the S\&P 500 index} & 0.9941\tabularnewline
dy & {\small{}Dividend Yield: the difference between the log of the 12-month
moving sum dividends and the log of lagged the S\&P 500 index} & 0.9941\tabularnewline
ep & {\small{}Earning Price Ratio: the difference between the log of the
12-month moving sum earnings and the log of the S\&P 500 index} & 0.9904\tabularnewline
tms & {\small{}Term Spread: the difference between the long-term government
bond yield and the Treasury Bill rate} & 0.9576\tabularnewline
dfy & {\small{}Default Yield Spread: the difference between Moody\textquoteright s
BAA- and AAA-rated corporate bond yields} & 0.9717\tabularnewline
dfr & {\small{}Default Return Spread: the difference between the returns
of long-term corporate bonds and long-term government bonds} & -0.0735\tabularnewline
bm & {\small{}Book-to-Market Ratio: the ratio of the book value to market
value for the Dow Jones Industrial Average} & 0.9927\tabularnewline
tbl & {\small{}Treasury Bill Rates: the 3-month Treasury Bill rates} & 0.9905\tabularnewline
ltr & {\small{}Long-Term Return: the rate of returns of long-term government
bonds} & 0.0500\tabularnewline
ntis & {\small{}Net Equity Expansion: the ratio of the 12-month moving sums
of net issues by NYSE listed stocks over the total end-of-year market
capitalization of NYSE stocks} & 0.9778\tabularnewline
svar & {\small{}Stock Variance: the sum of the squared daily returns on the
S\&P 500 index} & 0.4714\tabularnewline
infl & {\small{}Inflation: the log growth of the Consumer Price Index (all
urban consumers)} & 0.4819\tabularnewline
\bottomrule
\end{tabular}
\end{table}

\clearpage
\singlespacing
\bibliographystyle{chicago}
\bibliography{ref}




\end{document}